\documentclass[12pt]{article}
\usepackage{amsmath, amssymb, graphicx, booktabs}
\usepackage[margin=1in]{geometry}
\usepackage{setspace}
\usepackage{xcolor}
\usepackage[normalem]{ulem}

\usepackage{float}       
\usepackage{booktabs}    
\usepackage{caption}     

\usepackage{float}
\usepackage{booktabs}

\usepackage{parskip}

\usepackage{amsmath,amssymb,bm}

\usepackage{graphicx}
\usepackage{booktabs}
\usepackage{multirow}
\usepackage{array}
\usepackage{tabularx}
\usepackage{caption}

\usepackage{listings}
\usepackage{xcolor}
\definecolor{codekw}{RGB}{0,64,128}
\definecolor{codecomment}{RGB}{0,128,64}
\definecolor{codestring}{RGB}{128,0,128}
\definecolor{codebg}{RGB}{245,245,245}
\lstdefinelanguage{Rcode}{
  keywords={library,list,function,TRUE,FALSE,NULL,NA,seq,c,set.seed,cat,head,
            print,which,min,max,log,exp,pmax,pmin,rexp,runif,cumsum,ifelse,
            data.frame,round,as.integer,saveRDS,readRDS,str,paste,setdiff},
  sensitive=true,
  morecomment=[l]{\#},
  morestring=[b]",
}
\makeatletter
\renewcommand\@cite[2]{(#1\if@tempswa , #2\fi)}
\makeatother
\newcommand{\citep}[1]{\cite{#1}}
\usepackage[colorlinks=true,linkcolor=red,citecolor=blue,urlcolor=blue]{hyperref}

\newcounter{proposition}
\newcounter{lemma}
\newcounter{remark}
\newenvironment{proposition}[1][]{
  \refstepcounter{proposition}
  \par\medskip\noindent\textbf{Proposition~\theproposition\if\relax\detokenize{#1}\relax\else\ (#1)\fi.}\enspace\ignorespaces
}{
  \par\medskip
}

\newenvironment{proof}{
  \par\noindent\textit{Proof.}\enspace
}{
  \hfill$\square$\par\medskip
}

\newcommand{\IA}{\mathrm{IA}}
\newcommand{\FA}{\mathrm{FA}}
\newcommand{\obs}{\mathrm{obs}}
\newcommand{\Bw}{B_W}
\newcommand{\Fw}{\bar F_W}

\DeclareMathOperator{\Gammadist}{Gamma}

\providecommand{\keywords}[1]{\noindent\textbf{Keywords: } #1}
\definecolor{waterblue}{RGB}{0,170,220}
\newcommand{\revadd}[1]{{\color{waterblue}#1}}

\usepackage{listings}
\usepackage{xcolor}
\lstdefinestyle{rstyle}{
  language=R,
  basicstyle=\ttfamily\small,
  columns=fullflexible,
  keepspaces=true,
  showstringspaces=false,
  frame=single,
  breaklines=true,
  tabsize=2,
  keywordstyle=\color{blue!70!black}\bfseries,
  commentstyle=\color{black!50},
  stringstyle=\color{green!40!black}
}

\usepackage{tikz}
\usetikzlibrary{arrows.meta,positioning,calc,decorations.pathreplacing}

\usepackage{booktabs}
\usepackage{multirow}
\usepackage{graphicx}

\title{Interim Monitoring as an Information–Time Alignment Problem: The WCR Framework for Time-to-Event Trials}
\author{Haitao Pan \\ Zhongheng Cai\\
Department of Biostatistics, St. Jude Children's Research Hospital}
\date{\today}

\begin{document}

\maketitle

\begin{abstract}

Interim monitoring in time-to-event (TTE) trials is fundamentally an information--timing problem. Existing interim paradigms typically rely on proxies for inferential information, such as event counts or sample size, rather than directly controlling follow-up maturity itself. Event-driven designs align interim analyses with event accumulation but may incur substantial and unpredictable calendar delays in sparse-event settings. Enrollment-driven designs provide operational predictability but may trigger analyses based on immature follow-up and unstable inferential maturity. These limitations become particularly consequential in rare-disease and long-horizon TTE trials, where the same interim calendar time may correspond to markedly different effective information horizons depending on accrual dynamics, censoring structure, and estimand definition.

We propose the Window-Cohort with Calibrated Follow-Up Requirement (WCR) framework, a maturity-based interim monitoring approach that jointly aligns information and calendar time by directly parameterizing follow-up maturity. The proposed design defines interim timing through two quantities: a locked cohort size $N_W$ and a calibrated post-lock follow-up requirement $X$, yielding the interim trigger
\[
T_{\IA} = a_{(N_W)}  + X,
\]
where $a_{(N_W)} $ denotes the enrollment time of the $N_W$-th patient. Unlike conventional paradigms, the proposed framework treats interim timing as an estimand-aligned statistical design object rather than as a consequence of event accumulation or enrollment progression alone.

The framework further distinguishes restricted and unrestricted follow-up regimes corresponding to landmark survival and proportional hazards--based estimands, respectively, thereby clarifying how estimands determine the effective information horizon at interim. Design parameters, including $N_W$, $X$, total sample size, and decision thresholds, are calibrated through constrained optimization procedures that enforce type I error and power requirements while balancing operational efficiency, calendar predictability, and patient exposure during decision-lag periods.

Simulation studies show that the proposed framework meets target type~I error and power under its calibration model while improving calendar interpretability relative to conventional event-driven and enrollment-driven approaches. The methodology is implemented in the open-source \textsf{R} package \texttt{WCRBayesDesign}, available on CRAN, facilitating reproducible calibration and transparent design evaluation. More broadly, the proposed framework reframes interim monitoring in TTE trials as an information--time alignment problem and highlights follow-up maturity as a central determinant of inferential information in modern trials with long-horizon endpoints.

\end{abstract}


\section{Introduction}

Interim monitoring in time-to-event (TTE) clinical trials is fundamentally an information--timing problem. Decisions must be made early enough to remain operationally meaningful, yet sufficiently late to ensure that the available data contain adequate inferential information. Historically, these two objectives have been addressed indirectly through event accumulation or enrollment progression. Event-driven designs trigger analyses once a pre-specified number of events has occurred, whereas enrollment-driven designs trigger analyses once a target number of patients has been enrolled. Although these paradigms have become standard in modern clinical trial practice, both rely on proxies for inferential information rather than directly controlling the maturity structure of the underlying follow-up data.

This distinction becomes increasingly consequential in contemporary TTE trials characterized by long-horizon endpoints, sparse event processes, delayed treatment effects, slow accrual, or substantial administrative censoring. In such settings, event counts and sample size may become progressively disconnected from the actual maturity of inference. Consequently, interim analyses may no longer align with clinically actionable decision-making. Under some circumstances, event-driven monitoring may delay interim analyses until most patients have already been enrolled, thereby limiting the practical utility of futility monitoring. Conversely, enrollment-driven monitoring may produce analyses dominated by immature follow-up and heavy administrative censoring, yielding unstable inferential maturity across realizations. These limitations are not merely operational inconveniences. Rather, they reflect a structural mismatch between conventional interim monitoring paradigms and the information structure of modern TTE trials.

The limitations of this disconnect become especially pronounced in rare pediatric oncology and immunotherapy settings. Many contemporary pediatric solid-tumor studies evaluate long-horizon endpoints such as 2-year event-free survival (EFS), progression-free survival (PFS), duration of response (DOR), or milestone survival probabilities, while simultaneously operating under severe constraints in accrual and event availability. In such settings, event-driven monitoring may become operationally fragile because interim decisions occur only after prolonged calendar delays and after substantial patient exposure has already occurred. Enrollment-driven monitoring, while operationally predictable, may instead generate analyses with highly variable follow-up maturity depending on realized accrual dynamics. The resulting disconnect between operational timing and inferential maturity becomes increasingly problematic as TTE endpoints become longer, rarer, and more administratively censored.

More fundamentally, existing interim monitoring paradigms implicitly assume that event counts or enrollment counts adequately represent inferential information. However, in TTE settings, inferential information is driven primarily by the maturity structure of follow-up rather than by sample size or event accumulation alone. The same enrollment milestone may correspond to dramatically different information maturity depending on accrual dynamics, censoring patterns, and estimand definition. Similarly, the same event threshold may be reached at vastly different calendar times depending on the underlying event process. Existing paradigms therefore control either timing or information indirectly, but not both simultaneously.

This issue becomes even more important when the estimand itself depends on the effective follow-up horizon. Landmark survival estimands naturally induce horizon-restricted inference, whereas proportional hazards--based estimands use unrestricted follow-up extending to the full available observation window. Consequently, the same interim calendar time may correspond to fundamentally different effective information sets under different estimands. In this sense, interim timing is not purely an operational construct but an estimand-dependent inferential mechanism. These observations suggest that interim timing should not merely be inherited from event accumulation or enrollment progression. Rather, interim timing should itself be treated as a statistical design object that jointly governs information maturity, operational feasibility, and patient exposure.

Several strands of literature are relevant to this problem, although each addresses only part of the underlying information--time tension. Classical event-driven survival monitoring methods, including group sequential log-rank procedures and information-fraction frameworks, align interim analyses with event accumulation and thereby maintain strong inferential calibration under proportional hazards assumptions \revadd{\citep{obrien1979,lan1983}}. These approaches have formed the foundation of modern survival trial monitoring. However, because interim timing remains anchored to stochastic event realization, substantial calendar-time uncertainty may arise in sparse-event or long-horizon settings.

A second strand of literature focuses on Bayesian and adaptive monitoring procedures for TTE endpoints. Bayesian survival monitoring approaches and adaptive phase II survival designs have improved inferential flexibility under delayed outcomes, sparse events, and incomplete follow-up \revadd{\citep{wu2020twostage,wu2021bayesian,zhou2020bop2tte}}. Related methods incorporate posterior predictive probabilities, adaptive futility monitoring, or likelihood-based updating procedures to support decision-making under partial information. Nevertheless, in most existing approaches, interim timing remains fundamentally tied to event realization or enrollment progression. As a result, the underlying disconnect between timing and follow-up maturity remains unresolved.

A third strand of literature addresses delayed outcomes and pending responses in early-phase oncology studies. Methods such as time-to-event CRM procedures, TOP-type designs, and partial-information monitoring approaches explicitly recognize that outcomes may remain unobserved at interim analysis and attempt to incorporate partially observed follow-up into inference \revadd{\citep{cheung2000,lin2020top}}. These methods improve efficiency under delayed endpoint settings and reduce the need to suspend accrual while awaiting complete outcomes. However, follow-up maturity itself is typically treated as a secondary consequence of enrollment and event processes rather than as a primary design parameter governing interim timing.

A fourth strand of literature recognizes operational concerns associated with prolonged delays in event accumulation and proposes calendar-based safeguards or reporting backstops. Such approaches acknowledge the importance of operational feasibility, data monitoring, and timely reporting in modern TTE trials \revadd{\citep{fda2006dmc,ich2019e9r1}}. However, these procedures generally address reporting delays rather than the deeper inferential structure linking information maturity to interim timing.

Although these strands of literature emerged from different motivations, they share a common feature: existing approaches typically control proxies for inferential information rather than directly parameterizing follow-up maturity itself. Consequently, interim timing remains either operationally unpredictable, inferentially unstable, or estimand-disconnected in long-horizon TTE settings. This persistent limitation motivates the need for interim monitoring frameworks that directly control follow-up maturity and explicitly align information with calendar time.

To address this gap, we propose the Window-Cohort with Calibrated Follow-Up Requirement (WCR) framework for interim monitoring in single-arm TTE trials. The proposed framework introduces follow-up maturity as an explicit design parameter through two jointly specified quantities: a locked cohort size $N_W$ and a minimum follow-up requirement $X$. Interim analysis is triggered at
\[
T_{\IA} = a_{(N_W)}  + X,
\]
where $a_{(N_W)} $ denotes the enrollment time of the $N_W$-th patient. Under this construction, the analysis cohort is first fixed and then followed for a calibrated waiting period before interim evaluation occurs. Accrual may continue during this waiting window, but patients enrolled after cohort lock do not contribute to the interim decision.

Conceptually, the proposed framework represents a third interim monitoring paradigm distinct from both event-driven and enrollment-driven designs. Unlike event-driven approaches, WCR does not wait for stochastic event accumulation to determine when interim decisions may occur. Unlike enrollment-driven approaches, WCR does not assume that enrollment alone determines inferential maturity. Instead, WCR directly parameterizes follow-up maturity itself and thereby jointly constrains both timing and information.

The proposed framework also clarifies that effective information maturity is inherently estimand-dependent. Under landmark survival estimands, interim analyses become naturally restricted to the clinically meaningful horizon $\tau$, whereas proportional hazards--based estimands induce unrestricted follow-up usage extending to the entire available observation window. This distinction demonstrates that effective information maturity depends jointly on the interim timing mechanism and the estimand structure itself.

An additional feature of the WCR framework is the explicit characterization of ethical burden associated with interim timing mechanisms. Under sparse-event settings, prolonged waiting periods prior to interim decisions may expose additional patients to potentially ineffective therapy before futility can be declared. In conventional monitoring paradigms, this burden is often treated as incidental or unavoidable. In contrast, the WCR framework makes this burden explicit, measurable, and designable through quantities such as waiting-window enrollment burden and cumulative follow-up exposure during the decision-lag period. This allows trade-offs among inferential maturity, operational efficiency, and patient exposure to become transparent components of design calibration.

The current work is motivated in part by ongoing rare pediatric oncology studies such as RMS2021 (NCT06023641), a single-arm phase II rhabdomyosarcoma trial evaluating a long-horizon EFS endpoint under slow accrual and sparse event accumulation. In such settings, conventional event-driven monitoring may produce interim analyses only after many years of follow-up and after most patients have already been enrolled, whereas enrollment-driven approaches may yield analyses dominated by immature follow-up. Although RMS2021 provides a motivating example, the issues addressed here arise broadly across contemporary TTE studies involving rare diseases, immunotherapy, cell therapy, and long-term survival endpoints.

The proposed framework contributes both conceptually and methodologically. The main contributions are as follows. First, the paper reframes interim timing as an estimand-aligned maturity-control problem rather than as a rule inherited from event or enrollment counts. Second, it introduces the WCR trigger $(N_W,X)$, which separates cohort locking from the calibrated follow-up requirement. Third, it develops estimand-specific monitoring rules under restricted landmark follow-up and unrestricted HR-based follow-up. Fourth, it defines ethical burden metrics that quantify roll-on enrollment and follow-up exposure during the waiting window. Fifth, it provides simulation-based calibration and an open-source implementation in \texttt{WCRBayesDesign}.

Although Bayesian posterior monitoring provides the inferential engine in the current implementation, the conceptual contribution of the framework is not inherently Bayesian. Rather, the central contribution lies in the explicit parameterization of follow-up maturity and the joint alignment of information and calendar time. Consequently, the proposed framework is compatible with alternative inferential paradigms, including likelihood-based procedures, frequentist monitoring rules, hybrid Bayesian--frequentist approaches, and decision-theoretic adaptive monitoring strategies.

More broadly, the current work suggests that interim timing in TTE trials may need to be reinterpreted as an estimand-aligned statistical design object rather than merely a logistical consequence of event accumulation or enrollment progression. As modern clinical trials increasingly incorporate long-horizon endpoints, adaptive monitoring, sparse-event populations, and estimand-aware inference, the relationship between information and calendar time becomes progressively more central to statistical design itself. In this context, maturity-based monitoring frameworks may provide a useful direction for future development of adaptive TTE trial methodology.

The remainder of the paper is organized as follows. Section 2 introduces the conceptual and operational structure of the WCR framework. Section 3 develops the statistical inference and calibration procedures under both landmark and proportional hazards estimands. Section 4 presents implementation details and illustrative examples motivated by pediatric oncology applications. Section 5 evaluates operating characteristics through simulation studies and compares the proposed approach with conventional event-driven and enrollment-driven monitoring strategies. Section 6 concludes with discussion of theoretical implications, ethical considerations, and future directions for maturity-based interim monitoring in TTE trials.

\section{The WCR Design Framework}

\subsection{Conceptual Framework}

Interim monitoring in time-to-event (TTE) trials requires balancing two competing objectives: decisions should be made based on sufficiently mature outcome data, yet they must also occur within a predictable and operationally feasible time frame. As discussed in Section~1, existing design paradigms address these objectives only partially. Event-driven designs align interim analyses with statistical information through event accumulation, but provide limited control over when decisions occur. Enrollment-driven designs offer calendar predictability, but do not ensure that the available data are sufficiently informative at the time of analysis.

This limitation reflects a more fundamental issue: existing approaches control \emph{proxies} for information, such as event counts or sample size, rather than directly controlling the maturity of follow-up, which is the primary driver of information in TTE settings. As a result, interim timing is either unpredictable (under event-driven rules) or insufficiently informative (under enrollment-driven rules), especially in rare pediatric trials with slow accrual and long follow-up horizons.

We propose a conceptual shift in interim monitoring by introducing follow-up as an explicit design dimension. The key idea is to anchor interim analysis not to event accumulation or enrollment milestones alone, but to a \emph{pre-defined cohort with a calibrated follow-up requirement}. Specifically, the design is defined by two quantities: a cohort size $N_W$ and a minimum follow-up duration $X$ required for each patient in that cohort. Once the $N_W$-th patient is enrolled, the cohort is locked, and interim analysis is conducted only after all patients in the cohort have accrued at least $X$ units of follow-up.

This construction leads to a fixed-delay interim trigger that directly controls information maturity. By enforcing a minimum follow-up requirement, the design guarantees a lower bound on the informativeness of the interim data. At the same time, because the trigger depends only on enrollment up to $N_W$ and a fixed delay $X$, the timing of the interim analysis becomes predictable up to accrual variability prior to cohort lock. In this sense, the proposed approach jointly constrains both the timing and the informativeness of interim decisions.

\medskip

\noindent
\textbf{Conceptual comparison with existing designs.}
Figure~\ref{fig:timing_mechanisms} illustrates the distinction between three timing mechanisms. Event-driven designs tie interim timing to stochastic event accumulation, achieving information alignment at the cost of calendar uncertainty. Enrollment-driven designs fix calendar timing through enrollment milestones, but do not control follow-up maturity and therefore yield variable information at interim. In contrast, the proposed window-cohort design anchors interim timing to a fixed cohort and a calibrated follow-up requirement, thereby controlling both dimensions simultaneously. This figure highlights that event-driven and enrollment-driven designs each control only one dimension, whereas the proposed WCR design explicitly controls both.

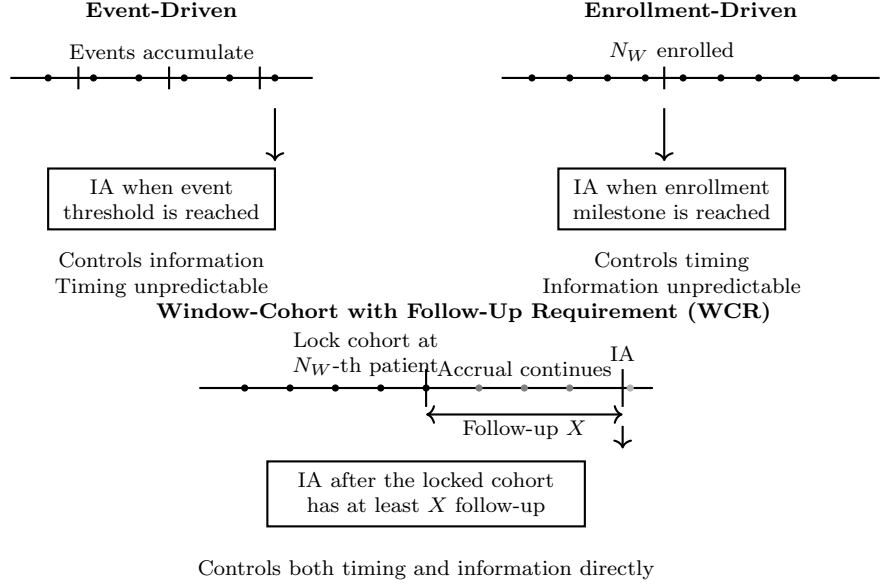
\begin{figure}[H]
\centering
\begin{tikzpicture}

\node[font=\scriptsize\bfseries] at (2.5,5.5) {Event-Driven};
\node[font=\scriptsize\bfseries] at (9.5,5.5) {Enrollment-Driven};

\draw[thick] (0.5,4.6) -- (4.5,4.6);
\foreach \x in {1.0,1.6,2.2,2.8,3.4,4.0} {
    \filldraw[black] (\x,4.6) circle (0.04);
}
\foreach \x in {1.4,2.6,3.8} {
    \draw[thick] (\x,4.45) -- (\x,4.75);
}
\node[font=\scriptsize] at (2.5,4.95) {Events accumulate};
\draw[thick,->] (4.0,4.2) -- (4.0,3.5);
\node[draw,thick,align=center,font=\scriptsize,minimum width=3cm,minimum height=0.8cm] 
at (2.5,3.0) {IA when event\\threshold is reached};
\node[align=center,font=\scriptsize] at (2.5,2.0) {Controls information\\Timing unpredictable};

\draw[thick] (7.0,4.6) -- (12.0,4.6);
\foreach \x in {7.4,7.9,8.4,8.9,9.4,9.9,10.4,10.9,11.4} {
    \filldraw[black] (\x,4.6) circle (0.04);
}

\draw[thick] (9.15,4.45) -- (9.15,4.75);

\node[font=\scriptsize] at (9.25,4.95) {$N_W$ enrolled};

\draw[thick,->] (9.15,4.2) -- (9.15,3.5);

\node[draw,thick,align=center,font=\scriptsize,
minimum width=3cm,minimum height=0.8cm] 
at (9.25,3.0) {IA when enrollment\\milestone is reached};

\node[align=center,font=\scriptsize] 
at (9.25,2.0) {Controls timing\\Information unpredictable};



\node[font=\scriptsize\bfseries] at (6.5,1.5) {Window-Cohort with Follow-Up Requirement (WCR)};

\draw[thick] (3.0,0.5) -- (9.0,0.5);

\foreach \x in {3.6,4.2,4.8,5.4,6.0} {
    \filldraw[black] (\x,0.5) circle (0.04);
}

\foreach \x in {6.7,7.3,7.9} {
    \filldraw[gray] (\x,0.5) circle (0.04);
}

\foreach \x in {8.7} {
    \filldraw[gray!60] (\x,0.5) circle (0.04);
}

\draw[thick] (6.0,0.25) -- (6.0,0.75);
\node[font=\scriptsize,align=center] at (5.2,0.95) {Lock cohort at\\$N_W$-th patient};

\draw[thick] (8.6,0.25) -- (8.6,0.75);
\node[font=\scriptsize] at (8.6,0.95) {IA};

\draw[thick,<->] (6.0,0.15) -- (8.6,0.15);
\node[font=\scriptsize] at (7.3,-0.05) {Follow-up $X$};

\node[font=\scriptsize,align=center] at (7.3,0.75) {Accrual continues};

\draw[thick,->] (8.6,0.0) -- (8.6,-0.3);

\node[draw,thick,align=center,font=\scriptsize,minimum width=4.2cm,minimum height=0.8cm] 
at (6.0,-0.9) {IA after the locked cohort\\has at least $X$ follow-up};

\node[align=center,font=\scriptsize] at (6.0,-1.9) {Controls both timing and information directly};
\end{tikzpicture}
\caption{
Conceptual comparison of three interim monitoring mechanisms in time-to-event trials. 
Top left: event-driven designs trigger interim analysis when a target number of events is observed, aligning decisions with statistical information but leading to unpredictable calendar timing. 
Top right: enrollment-driven designs trigger interim analysis after a fixed number of patients are accrued, ensuring predictable timing but yielding variable information depending on follow-up maturity. 
Bottom: the proposed window-cohort with follow-up requirement (WCR) design anchors interim analysis to a fixed cohort of size $N_W$ and a minimum follow-up duration $X$, thereby directly controlling both the timing and the informativeness of interim decisions.
}
\label{fig:timing_mechanisms}
\end{figure}

A complementary perspective is provided in Figure~\ref{fig:conceptual_anchor_space}, which positions these designs in a two-dimensional space defined by calendar predictability and information maturity. Existing approaches occupy opposing corners of this space, while the proposed design targets the upper-right region by jointly controlling both attributes.

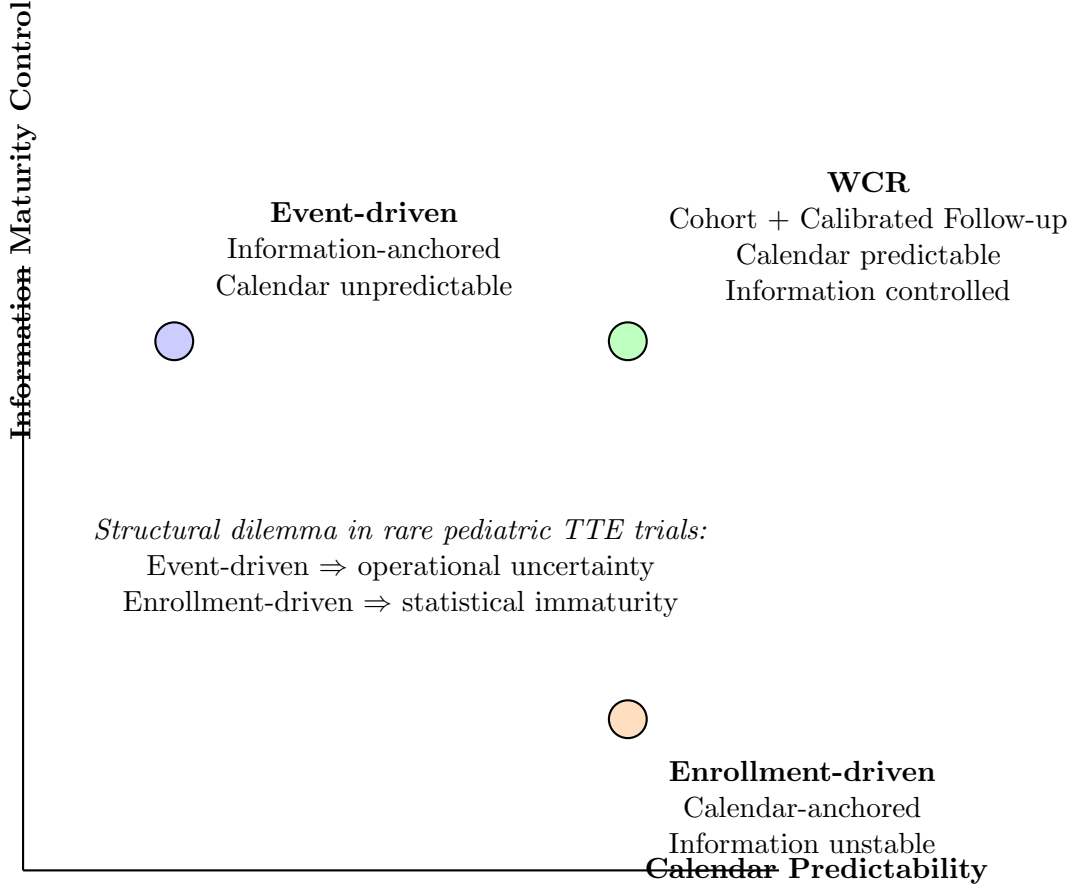
\begin{figure}[H]
\centering
\begin{tikzpicture}[
    font=\small,
    >=Latex,
    axis/.style={thick},
    point/.style={circle, draw, thick, minimum size=5mm, inner sep=0pt},
    labelstyle/.style={align=center}
]

\draw[axis] (0,0) -- (10,0);
\draw[axis] (0,0) -- (0,8);

\node at (10.5,0) {\textbf{Calendar Predictability}};
\node[rotate=90] at (0,8.6) {\textbf{Information Maturity Control}};

\node[point, fill=blue!20] (event) at (2,7) {};
\node[labelstyle, above right=3mm of event] 
{\textbf{Event-driven}\\
Information-anchored\\
Calendar unpredictable};

\node[point, fill=orange!25] (enroll) at (8,2) {};
\node[labelstyle, below right=3mm of enroll] 
{\textbf{Enrollment-driven}\\
Calendar-anchored\\
Information unstable};

\node[point, fill=green!25] (wcr) at (8,7) {};
\node[labelstyle, above right=3mm of wcr] 
{\textbf{WCR}\\
Cohort + Calibrated Follow-up\\
Calendar predictable\\
Information controlled};

\node[align=center] at (5,4)
{\textit{Structural dilemma in rare pediatric TTE trials:}\\
Event-driven $\Rightarrow$ operational uncertainty\\
Enrollment-driven $\Rightarrow$ statistical immaturity};

\end{tikzpicture}
\caption{Conceptual positioning of interim timing anchors. Event-driven designs align decisions with information but sacrifice calendar predictability. Enrollment-driven designs fix calendar timing but detach decisions from information maturity. The Window-Cohort with Calibrated Follow-Up Requirement (WCR) aligns both dimensions by anchoring interim timing to a locked cohort with a calibrated minimum follow-up constraint.}
\label{fig:conceptual_anchor_space}
\end{figure}

\medskip

\noindent
\textbf{Operational interpretation.}
The proposed framework can be viewed as a “fix-then-wait” strategy: first fix the analysis cohort at a pre-specified size, then allow follow-up to accrue in a controlled manner before conducting the interim analysis. Accrual may continue during this waiting window, but patients enrolled after cohort lock are excluded from interim analysis and carried forward to the final analysis. This separation ensures that the interim information set is fixed and amenable to calibration, while preserving overall enrollment efficiency.

\medskip

\noindent
\textbf{Design implications.}
By treating follow-up as a primary design parameter, the proposed framework transforms interim timing from an emergent property into a controllable design feature. This shift enables ex ante calibration of operating characteristics, including type I error, power, and interim stopping behavior, under explicitly defined information conditions. It also makes transparent the ethical trade-offs associated with delayed decision-making, such as the exposure of additional patients during the waiting window, which can be quantified and incorporated into design optimization.

The following sections formalize this framework, introduce notation and data structures, and develop procedures for calibration, evaluation, and implementation.

\subsection{Operational Definition of the WCR Design}

We now formalize the proposed window-cohort with calibrated follow-up requirement (WCR) design. The key components of the design are a pre-specified cohort size $N_W$ and a minimum follow-up requirement $X$. In this section, $X$ denotes the follow-up requirement induced by the design; its numerical value is selected by the calibration procedure in Section~\ref{sec:calibration-impl}.

Let $a_i$ denote the enrollment time of patient $i$, and let $a_{(N_W)}$ be the enrollment time of the $N_W$-th patient. We assume $N \geq N_W$ and resolve any enrollment-time ties by the preassigned patient index used in the trial database. The cohort of the first $N_W$ enrolled patients is referred to as the \emph{analysis cohort}. Once this cohort is formed, it is fixed for the purpose of interim analysis.

For design calibration and operating-characteristic evaluation, enrollment times are generated from a prespecified accrual model, taken in the RMS2021 example to be a homogeneous Poisson process with rate $r$. Event times $T_i$ are assumed to be independent and identically distributed from the working survival distribution under the scenario being evaluated and independent of the enrollment process. Censoring at the interim analysis is administrative, induced by the calendar analysis time; no independent dropout or informative censoring is included in the operating-characteristic calculations unless explicitly introduced in a sensitivity scenario.

The interim analysis (IA) is scheduled at
\[
T_{\mathrm{IA}} = a_{(N_W)} + X,
\]
where $X \ge 0$ is a pre-specified follow-up duration. That is, the interim analysis is conducted only after all patients in the analysis cohort have had at least $X$ units of potential follow-up time. The boundary case $X=0$ is interpreted as an enrollment-driven limiting case included for comparison and optimization rather than as a positive waiting-window design.

Accrual is allowed to continue during the interval $(a_{(N_W)}, T_{\mathrm{IA}})$. Patients enrolled after $a_{(N_W)}$ are not included in the interim analysis and are carried forward to the final analysis. We refer to these patients as \emph{roll-on patients}. This separation ensures that the interim data set is defined solely by the fixed analysis cohort and is therefore amenable to pre-trial calibration.

Let $T_i$ denote the true event time for patient $i$, and let $C_i^{\IA} = T_{\mathrm{IA}} - a_i$ be the administrative censoring time at interim. The observed time at interim is given by
\[
T_i^{\mathrm{obs,IA}} = \min(T_i, C_i^{\IA}),
\]
with event indicator
\[
\Delta_i^{\IA} = \mathbf{1}\{T_i \le C_i^{\IA}\}.
\]
The interim analysis is based on the data $\{(T_i^{\mathrm{obs,IA}}, \Delta_i^{\IA}): i = 1, \dots, N_W\}$. Zero observed events at interim are allowed; in that case the posterior update is driven by the accumulated follow-up time and the prior.

This construction yields a fixed-delay interim trigger that directly enforces a lower bound on follow-up maturity. In particular, for all patients in the analysis cohort, $C_i \ge X$ holds by design, ensuring that each contributes at least $X$ units of potential follow-up. As a result, the information available at interim is explicitly linked to the design parameters $(N_W, X)$, rather than being determined indirectly through event accumulation or accrual timing.

The pair $(N_W, X)$ therefore serves as the primary design parameters governing both the timing and the informativeness of the interim analysis. Their selection will be addressed through calibration procedures described in subsequent sections. This parameterization distinguishes the proposed design from existing approaches in which interim timing is determined indirectly through proxies for information.

\medskip

\noindent
\textbf{Operational interpretation.}
The WCR design can be interpreted as a three-step procedure. First, an analysis cohort of size $N_W$ is defined by fixing the first $N_W$ enrolled patients. Second, accrual continues while follow-up accumulates for the analysis cohort. Third, the interim analysis is triggered once all patients in the cohort have accrued at least $X$ units of follow-up. This ``fix-then-wait'' mechanism ensures that the interim data set is both well-defined and sufficiently mature for inference.

Figure~\ref{fig:wcr_operation} provides a schematic illustration of this process.

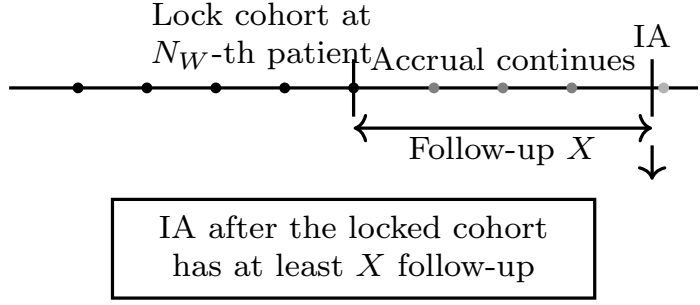
\begin{figure}[htbp]
\centering
\resizebox{0.8\textwidth}{!}{
\begin{tikzpicture}


\node[font=\scriptsize\bfseries] at (6.5,1.5) {Window-Cohort with Follow-Up Requirement (WCR)};

\draw[thick] (3.0,0.5) -- (9.0,0.5);

\foreach \x in {3.6,4.2,4.8,5.4,6.0} {
    \filldraw[black] (\x,0.5) circle (0.04);
}

\foreach \x in {6.7,7.3,7.9} {
    \filldraw[gray] (\x,0.5) circle (0.04);
}

\foreach \x in {8.7} {
    \filldraw[gray!60] (\x,0.5) circle (0.04);
}

\draw[thick] (6.0,0.25) -- (6.0,0.75);
\node[font=\scriptsize,align=center] at (5.2,0.95) {Lock cohort at\\$N_W$-th patient};

\draw[thick] (8.6,0.25) -- (8.6,0.75);
\node[font=\scriptsize] at (8.6,0.95) {IA};

\draw[thick,<->] (6.0,0.15) -- (8.6,0.15);
\node[font=\scriptsize] at (7.3,-0.05) {Follow-up $X$};

\node[font=\scriptsize,align=center] at (7.3,0.75) {Accrual continues};

\draw[thick,->] (8.6,0.0) -- (8.6,-0.3);

\node[draw,thick,align=center,font=\scriptsize,
minimum width=4.2cm,minimum height=0.8cm] 
at (6.0,-0.9) {IA after the locked cohort\\has at least $X$ follow-up};

\node[align=center,font=\scriptsize] at (6.0,-1.9)
{Controls both timing and information maturity};

\end{tikzpicture}
}
\caption{
Schematic illustration of the window-cohort with follow-up requirement (WCR) design. 
The first $N_W$ patients define the analysis cohort, which is locked upon enrollment of the $N_W$-th patient. 
Accrual continues while follow-up accumulates for the cohort. 
The interim analysis is conducted only after all cohort patients have accrued at least $X$ units of follow-up. 
Patients enrolled after cohort lock are excluded from the interim analysis and carried forward to the final analysis.
}
\label{fig:wcr_operation}
\end{figure}

\subsection{Notation and Data Structure}

We now introduce notation for the data structure induced by the WCR design at the interim analysis.

Let $a_i$ denote the enrollment time of patient $i$, for $i = 1, \dots, N$, where $N$ is the total number of patients enrolled over the course of the trial. Let $T_i$ denote the true event time, defined with respect to a specified time-to-event endpoint (e.g., event-free survival or progression-free survival).

As defined in Section~2.2, the interim analysis is conducted at time
\[
T_{\mathrm{IA}} = a_{(N_W)} + X,
\]
where $a_{(N_W)}$ is the enrollment time of the $N_W$-th patient and $X$ is the minimum follow-up requirement for the analysis cohort.

\medskip

\noindent
\textbf{Analysis cohort and censoring structure.}
The interim analysis is based on the analysis cohort consisting of the first $N_W$ enrolled patients. For each patient $i = 1, \dots, N_W$, the administrative censoring time at interim is given by
\[
C_i^{\mathrm{IA}} = T_{\mathrm{IA}} - a_i.
\]
By construction, $C_i^{\mathrm{IA}} \ge X$ for all $i = 1, \dots, N_W$, ensuring a minimum level of follow-up for each patient in the analysis cohort.

The observed data at interim are therefore
\[
T_i^{\mathrm{obs,IA}} = \min(T_i, C_i^{\mathrm{IA}}), 
\qquad
\Delta_i^{\mathrm{IA}} = \mathbf{1}\{T_i \le C_i^{\mathrm{IA}}\},
\]
for $i = 1, \dots, N_W$. The interim analysis is conducted using the dataset
\[
\mathcal{D}_{\mathrm{IA}} = \{(T_i^{\mathrm{obs,IA}}, \Delta_i^{\mathrm{IA}}): i = 1, \dots, N_W\}.
\]

\medskip

\noindent
\textbf{Roll-on patients.}
Patients enrolled after time $a_{(N_W)}$ are referred to as roll-on patients. These patients are not included in $\mathcal{D}_{\mathrm{IA}}$ and therefore do not influence interim decisions. However, they are retained for inclusion in the final analysis. This separation ensures that the interim data structure is fixed and determined solely by the analysis cohort.

\medskip

\noindent
\textbf{Final analysis.}
Let $T_{\mathrm{FA}}$ denote the time of the final analysis, which may be defined according to a pre-specified stopping rule or study completion criteria. The final analysis uses all enrolled patients and their available follow-up, yielding the dataset
\[
\mathcal{D}_{\mathrm{FA}} = \{(T_i^{\mathrm{obs,FA}}, \Delta_i^{\mathrm{FA}}): i = 1, \dots, N\},
\]
where $T_i^{\mathrm{obs,FA}} = \min(T_i, C_i^{\mathrm{FA}})$ and $C_i^{\mathrm{FA}} = T_{\mathrm{FA}} - a_i$.

\medskip

\noindent
\textbf{Information structure under WCR.}
The key feature of the WCR design is that the interim data $\mathcal{D}_{\mathrm{IA}}$ is structured by the design parameters $(N_W, X)$ and the realized accrual times up to $a_{(N_W)}$. In particular, the design enforces a lower bound on follow-up duration for each patient in the analysis cohort, while allowing the number and timing of observed events to remain stochastic. Thus, WCR controls the minimum follow-up structure, whereas realized inferential information remains conditional on the event-time distribution and censoring process.

This distinguishes the WCR data structure from event-driven designs, where the analysis is triggered by the realization of event times, and from enrollment-driven designs, where the analysis is triggered by sample size without explicit control of follow-up maturity. Under WCR, the minimum follow-up maturity of the interim dataset is directly governed by the design parameters through the induced censoring mechanism. This representation highlights that, under WCR, the administrative follow-up structure is controlled through the design of administrative censoring rather than through event realization or accrual timing.

\subsection{Illustrative Example}

We illustrate the operation of the WCR design through a simple example that reflects the settings described in Section~1.

Consider a single-arm phase II trial with a time-to-event endpoint and slow accrual, as in the RMS2021 setting described earlier. Suppose the WCR design is specified with cohort size $N_W = 20$ and follow-up requirement $X = 6$ months.

Patients are enrolled sequentially over time. Once the 20th patient is enrolled at time $a_{(20)}$, the analysis cohort is defined and fixed. Accrual continues beyond this point, and additional patients are enrolled during the subsequent waiting period.

The interim analysis is conducted at time
\[
T_{\mathrm{IA}} = a_{(20)} + 6 \text{ months},
\]
at which point all patients in the analysis cohort have had at least 6 months of potential follow-up. The interim dataset therefore consists of 20 patients, each with a minimum follow-up duration of 6 months, subject to administrative censoring at $T_{\mathrm{IA}}$.

Patients enrolled after $a_{(20)}$ are not included in the interim analysis and do not influence the interim decision. However, they continue to accrue follow-up and are incorporated into the final analysis. In this way, the interim decision is based on a fixed and well-defined information set, while overall trial enrollment is not interrupted. For example, under a Poisson accrual rate $r$, the expected number of roll-on patients during this waiting window is approximately $rX$; at $r=5/12$ patients per month and $X=6$ months, this corresponds to about 2.5 patients.

\medskip

\noindent
\textbf{Comparison with alternative timing rules.}
To highlight the role of the WCR design, consider how the same trial might proceed under alternative interim monitoring strategies. Under an event-driven design, the interim analysis would be triggered by the occurrence of a target number of events among enrolled patients. In settings with low event rates, this may occur only after substantial calendar time has elapsed, potentially when a large fraction of the trial has already been accrued.

Under an enrollment-driven design, the interim analysis would be conducted immediately upon enrollment of the 20th patient. In this case, the follow-up duration for each patient would depend on the realized accrual pattern. If accrual is relatively rapid, many patients may have only limited follow-up at interim, resulting in a dataset with substantial administrative censoring and limited information.

In contrast, under the WCR design, the interim analysis occurs at a predictable delay following cohort completion and is based on data with controlled follow-up maturity. This ensures that each patient in the analysis cohort contributes a minimum amount of information, while preserving a transparent and interpretable relationship between design parameters and interim data.

\medskip

\noindent
\textbf{Interpretation.}
This example illustrates how the WCR design separates the determination of the analysis cohort from the accumulation of follow-up. By fixing the cohort first and then allowing follow-up to accrue, the design directly controls the information content of the interim dataset. At the same time, because the interim timing is defined relative to the cohort lock time, it remains predictable up to variability in accrual prior to $a_{(N_W)}$.

The resulting interim dataset is therefore both well-defined and design-driven, providing a foundation for calibration and evaluation procedures developed in subsequent sections. 

\subsection{Comparison with Existing Timing Mechanisms}

We compare the proposed WCR design with commonly used interim monitoring strategies in time-to-event trials, focusing on how each approach determines the timing of interim analysis and the corresponding level of statistical information.

\medskip

\noindent
\textbf{A unifying perspective.}
Interim monitoring mechanisms can be characterized by two key attributes: (i) how the timing of the interim analysis is determined, and (ii) how the informativeness of the interim data is controlled. In existing designs, these two objectives are typically addressed separately, leading to inherent trade-offs between calendar predictability and information maturity.

\medskip

\noindent
\textbf{Event-driven designs.}
In event-driven designs, interim analyses are triggered by the accumulation of a pre-specified number of events. This approach aligns decision timing with statistical information, as event counts serve as a proxy for the effective sample size under time-to-event models. However, because event occurrence is stochastic and depends on both underlying hazard rates and follow-up duration, the timing of interim analysis is inherently unpredictable. In settings with low event rates or long-horizon endpoints, this may result in substantial delays, with interim decisions occurring only after most patients have already been enrolled.

\medskip

\noindent
\textbf{Enrollment-driven designs.}
In enrollment-driven designs, interim analyses are conducted once a fixed number of patients have been accrued. This ensures calendar predictability and administrative simplicity. However, sample size alone does not determine the amount of information available in time-to-event settings. The maturity of follow-up and the number of observed events depend on the realized accrual pattern, leading to variability in the informativeness of the interim dataset. Consequently, the same design may yield substantially different operating characteristics across realizations.

\medskip

\noindent
\textbf{Hybrid rules.}
In practice, hybrid rules are sometimes used to balance these competing considerations, for example by triggering interim analysis after a fixed number of patients or a fixed amount of time, whichever comes first. While such rules may mitigate extreme delays or immaturity, they remain ad hoc and do not directly control the underlying driver of information. As a result, their operating characteristics are difficult to characterize and justify in a principled manner.

\medskip

\noindent
\textbf{The WCR design.}
The proposed WCR design differs fundamentally from these approaches by directly controlling follow-up maturity through the design parameters $(N_W, X)$. By fixing the analysis cohort and imposing a minimum follow-up requirement, WCR ensures that each patient contributes at least $X$ units of potential follow-up at interim. Patients who experience an event before $X$ still have $X$ units of potential follow-up by design, but their observed event-free follow-up is shorter than $X$. This establishes a direct link between design parameters and the minimum maturity structure of the interim dataset. At the same time, because the interim timing is defined relative to the cohort lock time, it remains predictable up to variability in accrual prior to $a_{(N_W)}$.

\medskip

\noindent
\textbf{Summary comparison.}
Table~\ref{tab:timing_comparison} summarizes the key distinctions among these approaches.

\begin{table}[H]
\centering
\resizebox{0.9\textwidth}{!}{
\begin{tabular}{lccc}
\hline
\textbf{Design} & \textbf{Interim Trigger} & \textbf{Timing Control} & \textbf{Information Control} \\
\hline
Event-driven 
& Event count 
& Unpredictable 
& Indirect (via events) \\

Enrollment-driven 
& Sample size 
& Predictable 
& Indirect (via follow-up variability) \\

Hybrid rules 
& Ad hoc combination 
& Partially predictable 
& Not explicitly controlled \\

WCR (proposed) 
& Cohort + follow-up $(N_W, X)$ 
& Predictable (conditional on accrual) 
& Direct (via follow-up requirement) \\
\hline
\end{tabular}
}
\caption{
Comparison of interim monitoring mechanisms in time-to-event trials. Existing approaches control either timing or information indirectly through proxies, whereas the WCR design directly controls follow-up maturity and thereby jointly constrains both timing and information.
}
\label{tab:timing_comparison}
\end{table}

Taken together, these comparisons highlight that the key distinction of the WCR design is not merely the use of a fixed delay, but the explicit parameterization of follow-up as a design variable. This allows interim timing and information maturity to be jointly controlled within a unified framework.

\begin{proposition}[Estimand-dependent effective data horizon at interim]
\label{prop:estimand-horizon-1}
Under the WCR trigger, the interim analysis occurs at the deterministic calendar time
\[
T_{\mathrm{IA}} = a_{(N_W)} + X.
\]
Let $t_{\mathrm{cal}} = T_{\mathrm{IA}}$ denote the interim calendar time.

Then the effective data horizon at interim is determined jointly by the estimand and the available follow-up:

\begin{enumerate}
    \item[(i)] For the landmark estimand $\theta = S(\tau)$, the effective data horizon is truncated at $\tau$, yielding a $\tau$-restricted analysis.
    
    \item[(ii)] For the PH-based hazard-ratio estimand, the effective data horizon extends to the full available follow-up up to $t_{\mathrm{cal}}$, without additional truncation.
\end{enumerate}

Thus, while the WCR trigger fixes the calendar timing of interim analysis, the estimand determines the effective information horizon.
\end{proposition}

\subsection{Ethical and Operational Considerations}
\label{sec:ethical}

A key feature of the WCR design is that interim timing is determined by a fixed cohort and a pre-specified follow-up requirement. While this structure ensures controlled information maturity, it also introduces a potential ethical trade-off associated with continued enrollment during the waiting window $(a_{(N_W)}, T_{\mathrm{IA}}]$. A patient enrolled exactly at $T_{\mathrm{IA}}$ is counted as enrolled before the decision but contributes zero interim-period follow-up.

\medskip

\noindent
\textbf{Decision-lag enrollment burden.}
Under the WCR design, patients may continue to be enrolled after the analysis cohort is fixed but before the interim analysis is conducted. If the treatment is ultimately ineffective, these additional patients are exposed to a potentially futile therapy without contributing to the interim decision. We refer to this phenomenon as the \emph{decision-lag enrollment burden}.

Formally, let
\[
\mathcal{I}_W = \{ i : a_i \in (a_{(N_W)},\; a_{(N_W)} + X] \},
\qquad
B_W^{\obs} = |\mathcal{I}_W|
\]
denote the realized number of patients enrolled during the waiting window. This quantity captures the extent of additional enrollment occurring before the interim decision is made.

\medskip

\noindent
\textbf{Follow-up burden and patient exposure.}
While $B_W$ quantifies the number of additional patients, it does not capture the extent of their exposure. To characterize the total follow-up burden, define
\[
F_W^{\obs} = \sum_{i \in \mathcal{I}_W}
\min\bigl(T_i,\; T_{\mathrm{IA}} - a_i,\; \tau\bigr),
\]
for the landmark estimand, where $T_i$ denotes the potential event time. For a
PH-based estimand using unrestricted follow-up, the $\tau$ cap is omitted. This
quantity represents the cumulative observed follow-up time accrued by patients
enrolled during the waiting window. A corresponding per-patient measure is
given by
\[
\bar{F}_W^{\obs} = \frac{F_W^{\obs}}{B_W^{\obs}},
\]
when $B_W^{\obs} > 0$; by convention, $F_W^{\obs}=\bar F_W^{\obs}=0$ when $B_W^{\obs}=0$.

These quantities provide complementary realized summaries of decision-lag burden: $B_W^{\obs}$ reflects the number of additional patients exposed, while $F_W^{\obs}$ and $\bar{F}_W^{\obs}$ capture the magnitude of exposure in terms of follow-up time. In the simulation tables we report the corresponding expectations, $B_W=E(B_W^{\obs})$, $F_W=E(F_W^{\obs})$, and $\bar F_W=E(\bar F_W^{\obs})$, under the stated operating scenario. These quantities are descriptive in the primary optimization but can also be imposed as design constraints, for example $E_{H_0}(B_W^{\obs})\le b_{\max}$ or $E_{H_0}(F_W^{\obs})\le f_{\max}$, when a protocol requires an explicit cap on decision-lag exposure.

\medskip

\noindent
\textbf{Design interpretation.}
Unlike in event-driven designs, where interim timing depends on stochastic event accumulation, or enrollment-driven designs, where interim timing is tied to sample size without explicit control of follow-up, the WCR design makes the waiting window explicit and pre-specified. As a result, the associated ethical burden is not only transparent but also quantifiable in terms of $(N_W, X)$ and the accrual process.

\medskip

\noindent
\textbf{Comparison across design paradigms.}
The nature of ethical burden differs across interim monitoring strategies. In event-driven designs, enrollment may continue until a sufficient number of events is observed, leading to potentially large and unpredictable numbers of additional patients treated before a decision is made. In enrollment-driven designs, interim analysis is triggered immediately upon reaching the target sample size, minimizing decision lag but potentially relying on immature follow-up data; thus, limited lag burden does not imply negligible ethical risk. In contrast, the WCR design introduces a fixed waiting window that balances these considerations: it ensures a minimum level of follow-up maturity at interim while allowing the associated enrollment burden to be explicitly characterized.

\begin{table}[H]
\centering
\resizebox{0.9\textwidth}{!}{
\begin{tabular}{lccc}
\hline
\textbf{Design} & \textbf{Interim Trigger} & \textbf{Decision Lag} & \textbf{Decision-Lag Burden} \\
\hline
Event-driven 
& Event count 
& Unpredictable 
& Variable, not directly controlled \\

Enrollment-driven 
& Sample size 
& Minimal 
& Limited lag burden; may use immature data \\

WCR (proposed) 
& $a_{(N_W)} + X$ 
& Pre-specified 
& Explicit, quantifiable, and controllable \\
\hline
\end{tabular}
}
\caption{
Comparison of ethical burden across interim monitoring strategies. 
The WCR design makes the decision-lag burden explicit and designable through $(N_W, X)$.
}
\end{table}

\medskip

\noindent
\textbf{Implications for design calibration.}
Because $B_W$, $F_W$, and $\bar{F}_W$ are determined by the design parameters and the assumed accrual process, they can be evaluated under different scenarios and incorporated into design selection. For example, constraints on expected enrollment burden or follow-up exposure under the null hypothesis may be imposed alongside traditional operating characteristics such as type I error and power. This enables a principled trade-off between statistical performance and ethical considerations.

\medskip

\noindent
\textbf{Interpretation.}
These considerations highlight that ethical burden is inherent in all interim monitoring strategies, but manifests in different forms. The distinguishing feature of the WCR design is that it renders the waiting-window burden explicit, measurable, and amenable to design-level control. This provides a transparent framework for balancing decision timeliness, information maturity, and patient exposure in time-to-event trials.

\section{Statistical Inference and Design Calibration}
\label{sec:inference}

\paragraph{Estimands under the WCR trigger.}
Under the same WCR timing rule, different estimands induce different data
usages at interim.

\begin{itemize}
  \item[(i)] \emph{Landmark (milestone) estimand at $\tau$}. In the single-arm
    setting, define
    \begin{equation}
      \theta = S(\tau).
    \end{equation}
    At interim, estimation of $\theta$ is based on follow-up truncated at
    $\tau$; that is, each patient's contribution is restricted to
    $\min(T_i, \tau)$. This yields a $\tau$-restricted analysis consistent
    with the primary time-horizon definition.

  \item[(ii)] \emph{PH-based hazard-ratio estimand relative to a reference
    curve}. Alternatively, assume a proportional hazards structure
    \begin{equation}
      S(t) = \{S_0(t)\}^{\delta},
    \end{equation}
    where $\delta$ is the parameter of interest relative to a reference
    survival curve $S_0(t)$. In this case, estimation at interim uses all
    available follow-up up to $T_{\IA}$ without truncation at $\tau$. Thus
    the analysis is unrestricted with respect to the follow-up horizon.
\end{itemize}

Although both estimands can be mapped to survival probability at $\tau$ for
interpretability, they correspond to distinct modeling assumptions and use
different effective data horizons at interim.

\begin{proposition}[Estimand determines effective data horizon under fixed
  interim timing]
  \label{prop:estimand-horizon-2}
  Under the WCR trigger, the interim analysis occurs at the deterministic
  calendar time
  \begin{equation*}
    T_{\IA} = a_{(N_W)} + X.
  \end{equation*}
  Let $t_{\mathrm{cal}} = T_{\IA}$ denote the interim calendar time.

  Then the effective data horizon at interim depends on the estimand:
  \begin{enumerate}
    \item[(i)] \textbf{Landmark estimand at $\tau$.} If the target parameter
      is $\theta = S(\tau)$, then each patient contributes follow-up
      restricted to $\min(T_i, \tau)$. The interim analysis is therefore
      $\tau$-restricted, regardless of whether $(t_{\mathrm{cal}} - a_i)$
      exceeds $\tau$.
    \item[(ii)] \textbf{PH-based hazard-ratio estimand.} If the target
      parameter $\delta$ is defined through
      \begin{equation*}
        S(t) = \{S_0(t)\}^{\delta},
      \end{equation*}
      then estimation at interim uses all available follow-up up to
      $t_{\mathrm{cal}}$. No truncation at $\tau$ is imposed, and the
      analysis is unrestricted with respect to the follow-up horizon.
  \end{enumerate}

  Thus, while the WCR trigger fixes the calendar timing of interim analysis,
  the estimand determines the effective information horizon and data usage
  at interim.
\end{proposition}

\begin{proof}
  For the landmark estimand, the observed contribution is
  $u_i=\min\{T_i,\tau,(t_{\mathrm{cal}}-a_i)_+\}$ as defined in
  Section~\ref{sec:milestone}; therefore follow-up beyond $\tau$ cannot enter
  either the event count or the log-survival sum. For the PH-based estimand,
  \eqref{eq:hr-obs} defines
  $u_i=\min\{T_i,(t_{\mathrm{cal}}-a_i)_+\}$, so all available follow-up up to
  the interim calendar time enters the transformed cumulative-hazard sum, with
  no $\tau$ truncation. The two likelihood constructions therefore use the same
  WCR calendar trigger but different estimand-induced data horizons.
\end{proof}

The two estimands differ in their reliance on structural assumptions: the
landmark estimand targets a time-specific survival probability and remains
valid without assuming a particular relationship between survival curves,
whereas the HR-based estimand provides a global comparison but relies on the
proportional hazards assumption to ensure interpretability and model coherence.

\subsection{Milestone survival estimand at the primary horizon \texorpdfstring{$\tau$}{tau}}
\label{sec:milestone}

\paragraph{Estimand.}
The primary estimand is the milestone survival probability at the
protocol-defined horizon $\tau$:
\begin{equation}
  \theta \;=\; S(\tau) \;=\; \Pr(T > \tau).
\end{equation}
Data entering the analysis are $\tau$-restricted: each patient contributes
\[
  u_i \;=\; \min\!\bigl\{T_i,\;\tau,\;(t_{\mathrm{cal}} - a_i)_+\bigr\},
  \qquad
  \delta_i \;=\; \mathbf{1}\{T_i \le u_i\},
\]
so that no information beyond the landmark horizon is used; $d = \sum_i \delta_i$
denotes the total observed event count. Here and below, $n$ denotes the number
of records in the analysis dataset: $n=N_W$ at interim and $n=N_T$ at the final
analysis.

\paragraph{Exponential working likelihood.}
We first describe an exponential working likelihood for the underlying event
time. Under this working model, the
likelihood for the hazard rate $\lambda > 0$ given the $\tau$-restricted
data is
\begin{equation}
  \label{eq:working-lik-lambda}
  L(\lambda) \;\propto\; \lambda^{d}
    \exp\!\bigl(-\lambda\cdot\mathrm{TTOT}\bigr),
  \qquad \mathrm{TTOT} = \sum_{i=1}^{n} u_i.
\end{equation}
The sufficient statistic $\mathrm{TTOT}$ (total observation time) aggregates all
individual follow-up contributions.

\paragraph{Reparametrisation to the log-survival parameter $\varphi$.}
Under the exponential working model, the landmark survival satisfies
$\theta = S(\tau) = e^{-\lambda\tau}$, so $\lambda = \varphi/\tau$ where
$\varphi = -\log\theta \in (0,\infty)$.  We call $\varphi$ the
\emph{log-survival parameter} at $\tau$.  Substituting $\lambda = \varphi/\tau$
into \eqref{eq:working-lik-lambda} gives
\begin{equation}
  \label{eq:working-lik-phi}
  L(\varphi) \;\propto\; \varphi^{d}
    \exp\!\Bigl(-\varphi\cdot\underbrace{\frac{\mathrm{TTOT}}{\tau}}%
    _{\widetilde{\Omega}}\Bigr),
\end{equation}
a Gamma kernel in $\varphi$.  Under the exponential working model,
$\varphi = \lambda\tau = H(\tau)$ equals the cumulative hazard at the
landmark; under proportional hazards relative to a null reference, the hazard
ratio equals $\varphi/\varphi_L$ where $\varphi_L = -\log\theta_L$.
However, the estimand itself is model-free: $\varphi$ is defined solely through the landmark
survival probability $\theta = e^{-\varphi} = S(\tau)$. The posterior update below should therefore be interpreted as a calibrated working-likelihood construction whose operating characteristics depend on the chosen reference curve $S_0$.

\paragraph{Prior and posterior.}
With conjugate prior $\varphi \sim \Gammadist(a_0, b_0)$ (shape--rate
parameterisation), the posterior is
\begin{equation}
  \label{eq:posterior-delta-landmark}
  (-\log\theta) \mid d,\,\widetilde{\Omega}
  \;\sim\; \Gammadist\!\bigl(a_0 + d,\; b_0 + \widetilde{\Omega}\bigr).
\end{equation}

Unless otherwise stated, the numerical examples and simulations use
$a_0=b_0=0.01$, matching the diffuse working prior reported in
Section~\ref{sec:setup}. This prior is placed on the log-survival parameter
$\varphi=-\log\theta$ for the landmark estimand; on the survival-probability scale it induces $E\{S(24)\}=0.955$, median approximately 1.00, central 95\% interval approximately $(0.009,1.000)$, and $\Pr\{S(24)<0.62\}=0.047$ for the RMS2021 null threshold. Thus the prior is not uniform on the survival scale, and prior sensitivity should be examined when few events are expected. Sensitivity to the working
baseline is evaluated separately in Section~\ref{sec:p2-pw-robustness}.

\paragraph{Posterior inference on $\theta$.}
Since $\theta = e^{-\varphi}$ is a decreasing function of $\varphi$, all
posterior quantities on $\theta$ follow from \eqref{eq:posterior-delta-landmark}.
Applying the Gamma moment generating function at $t = 1$ yields the
posterior mean in closed form:
\begin{equation}
  \label{eq:Etheta}
  \mathbb{E}[\theta \mid d,\widetilde{\Omega}]
  \;=\; \mathbb{E}\!\bigl[e^{-\varphi}\bigr]
  \;=\; \left(\frac{b_0 + \widetilde{\Omega}}{b_0 + \widetilde{\Omega} + 1}
        \right)^{\!a_0+d}.
\end{equation}
The key posterior probability entering both decision rules is
$\Pr(\theta < \theta_L \mid d,\widetilde{\Omega})$, the posterior mass below
the null threshold.  Since $\{\theta < \theta_L\} \Leftrightarrow
\{\varphi > -\log\theta_L\}$, this equals the upper tail of the Gamma posterior:
\begin{equation}
  \label{eq:Prtheta}
  \Pr(\theta < \theta_L \mid d,\widetilde{\Omega})
  \;=\; \Pr\!\bigl(\varphi > {-\log\theta_L} \mid d,\widetilde{\Omega}\bigr),
\end{equation}
evaluated from the Gamma survival function of \eqref{eq:posterior-delta-landmark}.

\paragraph{Generalisation to an arbitrary reference baseline.}
In practice, the reference baseline $S_0(\cdot)$ may be chosen to reflect
non-exponential hazard shapes (e.g.\ Weibull or piecewise-Weibull) informed
by historical data.  Following the working-likelihood constructions used in prior single-arm TTE monitoring designs \citep{wu2020twostage,wu2021bayesian}, the exponential working likelihood extends to an
arbitrary $S_0$ by replacing $\mathrm{TTOT}/\tau$ with the
\emph{normalised log-survival sum}
\begin{equation}
  \label{eq:psi-transform}
  \widetilde{\Omega} \;=\; \frac{\Omega}{H_0(\tau)}
               \;=\; \frac{1}{-\log\theta_L}\sum_{i=1}^{n} H_0(u_i),
  \qquad H_0(t) = -\log S_0(t),
\end{equation}
where the reference baseline is calibrated so that $H_0(\tau) = -\log\theta_L$.
Under an exponential $S_0$, one has $H_0(u_i) = \lambda_0 u_i$ and
$H_0(\tau) = \lambda_0\tau$, so $\widetilde{\Omega}$ reduces to
$\mathrm{TTOT}/\tau$, recovering \eqref{eq:working-lik-phi} exactly.  For a
general $S_0$, the unnormalised log-survival sum $\Omega = \sum_i H_0(u_i)$
plays the role of $\mathrm{TTOT}$, and normalising by $H_0(\tau)$ ensures
that the posterior rate parameter is $b_0 + \widetilde{\Omega}$ for any choice
of baseline.  The posterior form \eqref{eq:posterior-delta-landmark} is
unchanged; the reference model contributes to the analysis solely through
$\widetilde{\Omega}$.

\paragraph{Decision rules (landmark scale).}
Both decision rules operate on the same posterior tail probability
$\Pr(\theta < \theta_L \mid \text{data})$ from \eqref{eq:Prtheta},
evaluated at the null survival threshold $\theta_L$.

\emph{Interim analysis (IA).}
\begin{equation}
  \text{Stop for futility if }
  \Pr(\theta < \theta_L \mid \text{IA data}) \geq p_{\IA}.
\end{equation}
Otherwise, the trial continues to the final analysis.

\emph{Final analysis (FA).}
\begin{equation}
  \text{Declare efficacy if }
  \Pr(\theta < \theta_L \mid \text{full data}) \leq p_{\FA}.
\end{equation}
Both tail probabilities are evaluated via \eqref{eq:posterior-delta-landmark}
and \eqref{eq:Etheta}--\eqref{eq:Prtheta}.

\paragraph{Example 1: Exponential baseline.}
Let $S_0(t) = \exp(-\lambda_0 t)$ with $\lambda_0 = (-\log\theta_L)/\tau$.
Then $H_0(u_i) = \lambda_0 u_i = (-\log\theta_L)\,u_i/\tau$, giving
\[
  \widetilde{\Omega} \;=\; \frac{\sum_i u_i}{\tau}
               \;=\; \frac{\mathrm{TTOT}}{\tau},
  \qquad \mathrm{TTOT} = \sum_i u_i,
\]
and the posterior \eqref{eq:posterior-delta-landmark} specialises to
\[
  (-\log\theta) \mid d,\,\mathrm{TTOT}
  \;\sim\; \Gammadist\!\Bigl(a_0 + d,\; b_0 + \tfrac{\mathrm{TTOT}}{\tau}\Bigr).
\]
This recovers the standard exponential working-likelihood result used in prior Bayesian single-arm TTE monitoring designs \citep{wu2020twostage,wu2021bayesian}.

\paragraph{Example 2: Weibull baseline.}
Let $S_0(t) = \exp\{-(\lambda_0 t)^k\}$ with known shape $k > 0$ and
$\lambda_0 = (-\log\theta_L)^{1/k}/\tau$.
Then $H_0(u_i) = (\lambda_0 u_i)^k = (-\log\theta_L)\,(u_i/\tau)^k$, so
\[
  \widetilde{\Omega} \;=\; \sum_i \!\Bigl(\frac{u_i}{\tau}\Bigr)^{\!k},
\]
and the posterior is
\[
  (-\log\theta) \mid d,\,\widetilde{\Omega}
  \;\sim\; \Gammadist\!\Bigl(a_0+d,\;
    b_0 + \textstyle\sum_i (u_i/\tau)^k\Bigr).
\]
Setting $k = 1$ recovers Example~1.

\paragraph{Example 3: Piecewise Weibull baseline (two segments).}
Let $t_{\mathrm{break}}$ denote a clinically meaningful changepoint (e.g.\
end of induction therapy).  The piecewise Weibull survival function is
\[
  S_0(t) \;=\;
  \begin{cases}
    \exp\!\bigl\{-(\lambda_1 t)^{k_1}\bigr\},
    & t \le t_{\mathrm{break}}, \\[4pt]
    \exp\!\Bigl\{-(\lambda_1 t_{\mathrm{break}})^{k_1}
      -\bigl[(\lambda_2 t)^{k_2} - (\lambda_2 t_{\mathrm{break}})^{k_2}\bigr]
    \Bigr\},
    & t > t_{\mathrm{break}},
  \end{cases}
\]
so the cumulative hazard is
\[
  H_0(t) \;=\;
  \begin{cases}
    (\lambda_1 t)^{k_1},
    & t \le t_{\mathrm{break}}, \\[4pt]
    (\lambda_1 t_{\mathrm{break}})^{k_1}
      + (\lambda_2 t)^{k_2} - (\lambda_2 t_{\mathrm{break}})^{k_2},
    & t > t_{\mathrm{break}},
  \end{cases}
\]
where shapes $k_1, k_2$ and rates $\lambda_1, \lambda_2$ are calibrated to
historical data subject to $H_0(\tau) = -\log\theta_L$.  The normalised
log-survival sum is
\[
  \widetilde{\Omega} \;=\; \frac{1}{-\log\theta_L}
  \Biggl[
    \sum_{i:\,u_i \le t_{\mathrm{break}}}(\lambda_1 u_i)^{k_1}
    +\sum_{i:\,u_i > t_{\mathrm{break}}}
      \Bigl((\lambda_1 t_{\mathrm{break}})^{k_1}
            + (\lambda_2 u_i)^{k_2}
            - (\lambda_2 t_{\mathrm{break}})^{k_2}\Bigr)
  \Biggr].
\]
The posterior form \eqref{eq:posterior-delta-landmark} is unchanged; only
$\widetilde{\Omega}$ differs.  This specification accommodates hazard profiles
common in paediatric oncology, where the hazard shape changes at a
clinically meaningful time point such as the end of induction therapy.

\subsection{PH-based hazard-ratio estimand relative to a reference curve}
\label{sec:phestimand}

\paragraph{Estimand.}
The estimand is the hazard-ratio parameter $\delta > 0$ defined through
\begin{equation}
  S(t) = \{S_0(t)\}^{\delta}, \qquad t \geq 0,
\end{equation}
where $S_0(\cdot)$ is a known reference survival curve.  Equivalently,
$\delta = H(t)/H_0(t)$ for all $t \ge 0$; $\delta < 1$ indicates longer
survival than the reference.  For interpretability at the clinically
relevant horizon $\tau$, the landmark survival
$\theta = S(\tau) = \{S_0(\tau)\}^{\delta} = \theta_L^{\,\delta}$
maps $\delta$ to a clinically meaningful scale.

\paragraph{Effective data horizon at interim.}
Under the PH estimand, inference uses \emph{unrestricted} follow-up; each
patient contributes
\begin{equation}
  \label{eq:hr-obs}
  u_i \;=\; \min\!\bigl\{T_i,\;(t_{\mathrm{cal}} - a_i)_+\bigr\},
  \qquad
  \Delta_i \;=\; \mathbf{1}\{T_i \le u_i\},
\end{equation}
with no truncation at $\tau$.

\paragraph{Log-survival sum and posterior.}
The log-survival sum for the HR estimand is
\begin{equation}
  \label{eq:psi-transform-hr}
  \Omega \;=\; \sum_{i=1}^{n} H_0(u_i)
           \;=\; \sum_{i=1}^{n}\bigl(-\log S_0(u_i)\bigr),
\end{equation}
computed with unrestricted $u_i$ from \eqref{eq:hr-obs} rather than the
$\tau$-capped version of \eqref{eq:psi-transform}.  With conjugate prior
$\delta \sim \Gammadist(a_0, b_0)$, the posterior is
\begin{equation}
  \label{eq:posterior-delta-hr}
  \delta \mid D,\,\Omega \;\sim\; \Gammadist\!\bigl(a_0 + D,\; b_0 + \Omega\bigr),
\end{equation}
where $D = \sum_i \Delta_i$.  The posterior mean and variance are
\begin{equation}
  \label{eq:post-moments-hr}
  \mathbb{E}[\delta \mid D,\Omega]
  \;=\; \frac{a_0 + D}{b_0 + \Omega},
  \qquad
  \operatorname{Var}[\delta \mid D,\Omega]
  \;=\; \frac{a_0 + D}{(b_0 + \Omega)^2}.
\end{equation}
Unless otherwise stated, the HR analyses use the same numerical
hyperparameters $a_0=b_0=0.01$, now placed on the hazard-ratio parameter
$\delta$. Because $\varphi$ and $\delta$ live on different scales, this shared
shape--rate pair should be interpreted as a common weak-prior convention for
implementation rather than as an assertion of identical prior information on
the two estimands.

\paragraph{Decision rules (hazard-ratio scale).}
Let $\delta_{\mathrm{goal}}$ denote the null hazard-ratio level. In the
RMS2021 application, $\delta_{\mathrm{goal}}=1$, corresponding to the null
survival curve satisfying $S_0(\tau)=\theta_L$. Let
$p_{\IA}^{(\delta)}$ and $p_{\FA}^{(\delta)}$ denote the posterior
probability cutoffs.  Both decision rules operate on the same posterior
tail probability $\Pr(\delta \geq \delta_{\mathrm{goal}} \mid \text{data})$,
the posterior mass at or beyond the null level. Since smaller $\delta$ implies
longer survival than the reference, small posterior mass above
$\delta_{\mathrm{goal}}$ is evidence against the null.

\emph{Interim analysis (IA).}
\begin{equation}
  \text{Stop for futility if }
    \Pr\!\bigl(\delta \geq \delta_{\mathrm{goal}} \mid \text{IA data}\bigr)
    \geq p_{\IA}^{(\delta)}.
\end{equation}
Otherwise, the trial continues to the final analysis.

\emph{Final analysis (FA).}
\begin{equation}
  \text{Declare efficacy if }
    \Pr\!\bigl(\delta \geq \delta_{\mathrm{goal}} \mid \text{full data}\bigr)
    \leq p_{\FA}^{(\delta)}.
\end{equation}
Both tail probabilities are evaluated from the Gamma posterior
\eqref{eq:posterior-delta-hr}.

\paragraph{Remark.}
The horizon $\tau$ enters the HR estimand only as an interpretability anchor
via $\theta = \theta_L^{\,\delta}$.  Inference itself uses the full observed
follow-up, and the posterior \eqref{eq:posterior-delta-hr} has the same
conjugate Gamma structure as \eqref{eq:posterior-delta-landmark}, with
$\Omega$ computed without the $\tau$-cap.  Once $S_0(\cdot)$ is specified, no additional parametric form is needed for posterior computation; however, operating characteristics depend on the adequacy of that reference curve over the observed follow-up range.

\paragraph{Interpretation.}
The WCR framework separates three components of trial design: the timing of
data collection, the statistical model for inference, and the decision
criteria. While the interim timing and cohort definition are shared across
estimands, the inference and decision rules remain estimand-specific:
the landmark estimand uses a generalised working likelihood with
$\tau$-restricted follow-up (Section~\ref{sec:milestone}), whereas the
hazard-ratio estimand uses the same likelihood structure with unrestricted
follow-up. This separation allows a common temporal structure
to be combined with different inferential targets within a unified design
framework.

\subsection{Design calibration under the WCR framework}
\label{sec:calibration-impl}

Under the WCR framework, the interim timing parameters $(N_W, X)$ and
decision thresholds $(p_{\IA}, p_{\FA})$ jointly determine the statistical,
operational, and ethical performance of the design. Calibration is therefore
formulated as a structured search over the design space, guided by operating
characteristic constraints and optimization objectives.

The calibration requires explicit null and alternative scenarios. For the
landmark estimand, we write $H_0:\theta=\theta_L$ and
$H_1:\theta=\theta_{\mathrm{alt}}$. For the HR estimand, we write
$H_0:\delta=\delta_{\mathrm{goal}}$ and
$H_1:\delta=\delta_{\mathrm{alt}}$, where
$\delta_{\mathrm{goal}}=1$ and
$\delta_{\mathrm{alt}}=\log(\theta_{\mathrm{alt}})/\log\{\theta_L\}$ when the
reference curve is calibrated so that $S_0(\tau)=\theta_L$. In the RMS2021
application below, $\theta_L=0.62$, $\theta_{\mathrm{alt}}=0.80$,
$\alpha_{\mathrm{target}}=0.05$, $1-\beta_{\mathrm{target}}=0.80$,
$p_{H0}=0.40$, and $p_{H1}=0.10$. Type I error and power are evaluated by Monte
Carlo simulation under the working baseline used for calibration, with
additional Weibull-shape sensitivity, accrual-rate sensitivity, and reference-curve robustness analyses reported in Sections~\ref{sec:1A}, \ref{sec:1B}, and~\ref{sec:p2-pw-robustness}, respectively.

\paragraph{Overview.}
We adopt a three-step calibration strategy:
\begin{center}
  (i) screen feasible interim designs \quad
  (ii) optimize final-stage parameters \quad
  (iii) validate operating characteristics.
\end{center}
This procedure separates the control of interim behavior from the
optimization of final performance, while preserving coherence with the WCR
timing structure.

\paragraph{Step 1: Interim feasibility screening.}
We first identify interim configurations that provide adequate early
stopping behavior under the null hypothesis. Specifically, we search over
\begin{equation}
  (N_W, X) \in \mathcal{W} \times \mathcal{X},
\end{equation}
where $\mathcal{W}$ and $\mathcal{X}$ are prespecified candidate sets for
cohort size and follow-up requirement. For each $(N_W, X)$, the interim
Monte Carlo simulation under the working baseline produces the sampling
distribution of the posterior tail
$\Pr(\theta < \theta_L \mid \text{IA data})$ (landmark estimand) or
$\Pr(\delta \ge \delta_{\mathrm{goal}} \mid \text{IA data})$ (HR estimand)
under both $H_0$ and $H_1$.

The interim decision threshold $p_{\IA}$ is calibrated via a
one-dimensional grid search so that the futility stopping rule achieves a
desired early stopping probability under $H_0$. We retain configurations
satisfying
\begin{equation}
  \mathcal{D}_{\IA} = \Bigl\{(N_W, X, p_{\IA}) : p_{\text{stop},H_0} \geq p_{H0}\Bigr\},
\end{equation}
where $p_{\text{stop},H_0}$ denotes the probability of stopping at interim
under $H_0$.

Optionally, an additional constraint
\begin{equation}
  p_{\text{stop},H_1} \leq p_{H1}
\end{equation}
is imposed to limit false futility under the alternative hypothesis.

\paragraph{Step 2: Final-stage optimization.}
Given the feasible interim set $\mathcal{D}_{\IA}$, we optimize final-stage
parameters $(N_T, p_{\FA})$ for each candidate
$(N_W, X, p_{\IA}) \in \mathcal{D}_{\IA}$. The final decision rule uses
$\theta$ (landmark) or $\delta$ (HR) as defined in the corresponding
estimand subsection, with posterior tails evaluated through the Gamma
conjugate posteriors of \eqref{eq:posterior-delta-landmark} or
\eqref{eq:posterior-delta-hr}.

We retain designs satisfying the global operating characteristic
constraints
\begin{equation}
  \alpha(\mathbf{d}) \leq \alpha_{\text{target}}, \qquad
  \text{Power}(\mathbf{d}) \geq 1 - \beta_{\text{target}},
\end{equation}
and select an optimal design $\mathbf{d}^*$ according to a prespecified
objective function, such as
\begin{align}
  & \min \mathbb{E}_{H_0}[N_T] && \text{(expected sample size)}, \\
  & \min \max N_T && \text{(minimax design)},
\end{align}
or
\begin{equation}
  \min \mathbb{E}[T_{\FA}] \quad \text{(calendar-time efficiency)}.
\end{equation}

\paragraph{Step 3: Validation.}
The selected design
$\mathbf{d}^* = (N_T, N_W, X, p_{\IA}, p_{\FA})$ is evaluated using
independent simulation to obtain stable estimates of its operating
characteristics, including type I error, power, early stopping
probabilities, expected sample size, and calendar-time metrics.

If the design fails to meet the target criteria, the search procedure is
iterated by refining $\mathcal{W}$, $\mathcal{X}$, or the optimization
objectives.

\paragraph{Remarks.}
This calibration framework treats interim timing $(N_W, X)$ and decision
thresholds $(p_{\IA}, p_{\FA})$ as joint design parameters. By explicitly
separating interim feasibility, final-stage optimization, and validation,
the procedure ensures that the resulting WCR design achieves desired
statistical performance while maintaining control over calendar timing and
information maturity. The Weibull working baseline used for the landmark
estimand preserves the tractability of the exponential case (closed-form
Gamma posterior after the transformation~\eqref{eq:psi-transform}) while
accommodating non-constant hazards; the HR estimand retains the same
conjugate posterior structure through the generic reference-curve
transformation~\eqref{eq:psi-transform-hr}. Implementation-level details for constructing analysis datasets, applying \texttt{conduct()}, and reproducing simulation workflows are provided in Web Appendix~D.

\section{Introduction to the R Package}
\label{sec:intro}

\subsection{Clinical Background}
\label{sec:background}

Rhabdomyosarcoma (RMS) is the most common soft-tissue sarcoma in
children and adolescents, accounting for approximately 3--4\% of all
pediatric cancers and roughly 350 new diagnoses per year in the
United States \citep{mcevoy2023pediatric}.  The disease arises from
skeletal-muscle progenitor cells and can occur at virtually any
anatomical site, including the head and neck, genitourinary tract,
and extremities.  Two major histological subtypes are recognized:
embryonal RMS, which comprises about 57\% of cases and carries a
relatively favorable prognosis, and alveolar RMS, which is
associated with poorer outcomes.  Treatment is multimodal, combining
surgical resection, radiation therapy, and multi-agent chemotherapy;
the vincristine--dactinomycin--cyclophosphamide (VAC) backbone has
remained the standard of care for over three decades.

Risk stratification in RMS is determined by a combination of
clinical and molecular features, including disease stage, surgical
group, age, tumor size, and---most importantly---FOXO1 fusion status
\citep{hibbitts2019refinement,hettmer2022molecular}.  FOXO1 fusion-positive
tumors carry substantially worse outcomes than fusion-negative tumors,
with event-free survival (EFS) of 52\% versus 78\% for localized
disease.  Patients classified as intermediate risk constitute the
largest subgroup for which new agents must be evaluated, yet two-year
EFS in this group has plateaued at 60--65\% despite successive
treatment intensifications.  The Children's Oncology Group (COG)
Phase~III trial ARST0531 \citep{hawkins2018addition}, which randomized
448 intermediate-risk patients to VAC versus VAC alternating with
vincristine and irinotecan (VI), confirmed VAC as the reference
regimen: four-year EFS was almost 62\% on the VAC arm with no significant
improvement from the addition of VI.

The ongoing RMS2021 study (NCT06023641) at St.~Jude Children's
Research Hospital is a Phase~II trial evaluating the combination of
liposomal irinotecan with VAC for intermediate-risk patients, guided
by molecular risk stratification based on FOXO1 fusion status.
We anchor the null hypothesis to the ARST0531 benchmark, setting
$\theta_L = 0.62$ at the $\tau = 24$-month landmark, and adopt
$\theta_{\mathrm{alt}} = 0.80$---an 18-percentage-point improvement
deemed clinically meaningful by the investigators.  Projected accrual
is approximately 5 patients per year, with a target enrollment of
about 46 patients.  We apply the WCR design to this setting by jointly
calibrating the locked-cohort size $N_W$ and the additional follow-up
window $X$ after the cohort closes.  This construction gives the
interim analysis enough event information to support a futility
decision while keeping the monitoring time scale compatible with a
slow-accruing pediatric trial.  This section uses the
\textsf{WCRBayesDesign} R package to carry out this calibration for
the RMS2021 setting.

For reproducibility, the analyses in this revision refer to
\textsf{WCRBayesDesign} version~1.0.1. The package is available from the
author's r-universe repository,
\url{https://zhongheng-biostatistics.r-universe.dev/WCRBayesDesign}. It provides
the calibration and simulation functions used below, including
\texttt{find\_Nw\_pIA}, \texttt{two\_stage\_optimize\_design},
\texttt{oc\_two\_stage}, and \texttt{conduct()}. For a complete reproducibility archive, the source repository or archived release, R version, RNG settings, seeds, and session information should be reported alongside the package version.

\subsection{Common Clinical Setup}
\label{sec:setup}

The RMS2021 example is worked under both estimands supported by
the WCR design --- the landmark survival $\theta = S(\tau)$ and the
proportional-hazards contrast $\delta = \log S(\tau)/\log S_0(\tau)$
--- which share the clinical inputs summarised below.

\begin{center}
\begin{tabular}{lll}
\toprule
Parameter & Value & Source \\
\midrule
Landmark horizon $\tau$ & 24 months & RMS2021 primary endpoint \\
Null survival $\theta_L$ & 0.62 & ARST0531 VAC arm (observed) \\
Alternative survival $\theta_{\mathrm{alt}}$ & 0.80
  & RMS2021 protocol hypothesis \\
One-sided $\alpha$ & 0.05 & Regulatory convention \\
Target power $1-\beta$ & 0.80 & Regulatory convention \\
Accrual rate $r$ & $5/12$ patients/month (5/year) & RMS2021 projected \\
Gamma prior on $\lambda$ & $a_0 = b_0 = 0.01$ & Weakly informative \\
\bottomrule
\end{tabular}
\end{center}

\noindent
The follow-up window $X$ is searched over a grid spanning the full
range from a purely enrollment-driven interim ($X = 0$) to full
landmark maturity ($X = 24$):
\[
  \mathcal{X} = \{0,\, 3,\, 6,\, 9,\, 12,\, 15,\, 18,\, 21,\, 24\}
  \quad\text{(months)}.
\]
For each value of $X$, the optimization evaluates 20 candidate values
of $N_W$ and selects the smallest interim cutoff $p_{\IA}$ that
satisfies both the futility-sensitivity and futility-specificity
constraints described in Section~\ref{sec:landmark}.

The null survival function $S_0(t)$ is specified as a piecewise
Weibull model rather than a constant-hazard exponential, because the
ARST0531 data exhibit a clearly non-constant hazard.  The model
comprises two Weibull segments joined at $t_{\mathrm{break}} = 24$
months with continuity of the cumulative hazard:
\[
  S_0(t) =
  \begin{cases}
    \exp\!\bigl[-(\lambda_1 t)^{k_1}\bigr],
      & t \le 24, \\[4pt]
    \exp\!\bigl[-H_{24} - (\lambda_2 t)^{k_2} + (\lambda_2 \cdot 24)^{k_2}\bigr],
      & t > 24,
  \end{cases}
\]
where $H_{24} = (\lambda_1 \cdot 24)^{k_1}$ is the cumulative hazard
at the breakpoint.  We fit $S_0(t)$ to the reconstructed ARST0531
VAC-arm Kaplan--Meier curve by differential evolution (DEoptim),
followed by L-BFGS-B refinement, minimizing the maximum absolute
deviation from the reconstructed curve and constraining
$S_0(24) = 0.62$.  The estimated parameters are
\[
  k_1 = 1.889,\quad \lambda_1 = 0.02819,\quad
  k_2 = 0.07857,\quad \lambda_2 = 0.005901.
\]
Shape $k_1 > 1$ reflects an increasing hazard during the first two
years of treatment, while $k_2 \approx 0.08 \ll 1$ captures the
sharp hazard decline after induction is complete---consistent with
the clinical expectation that patients who survive past the treatment
phase face substantially lower risk of relapse.
Figure~\ref{fig:pw-fit} displays the fit.

\begin{figure}[ht]
\centering
\includegraphics[width=0.85\textwidth]{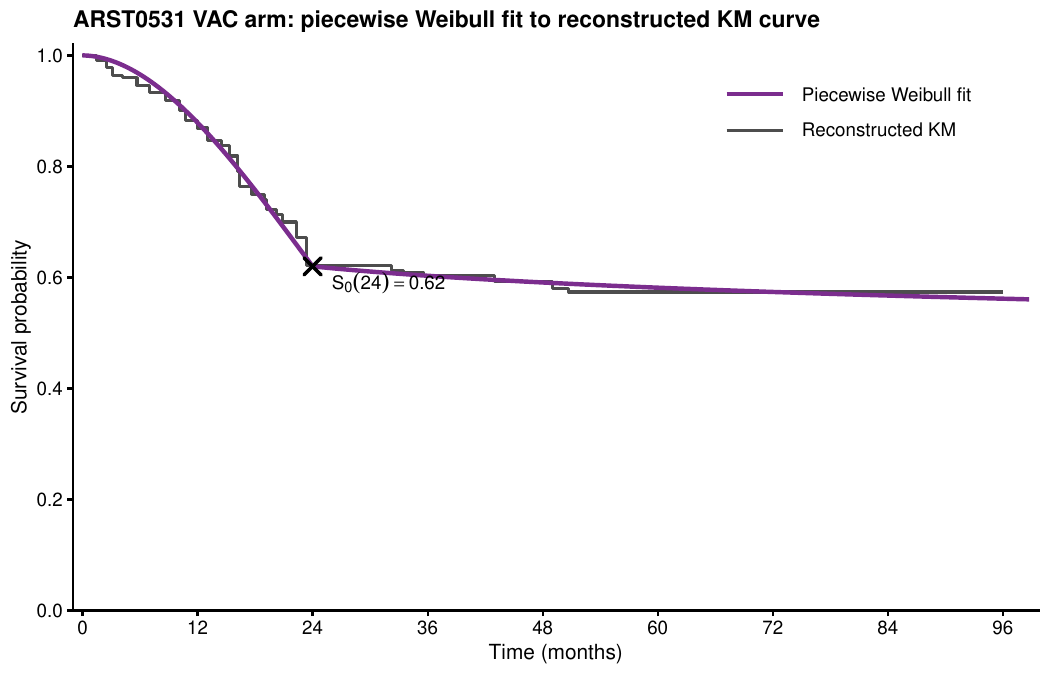}
\caption{Piecewise Weibull fit to the ARST0531 VAC-arm reconstructed
  Kaplan--Meier curve. The cross marks the anchored 2-year survival
  $S_0(24) = 0.62$. Shape $k_1 > 1$ before 24 months reflects increasing
  hazard during treatment; shape $k_2 \approx 0.08$ after 24 months
  captures the sharp hazard decline once induction is complete.}
\label{fig:pw-fit}
\end{figure}

Because the estimated post-24-month shape is very small, this tail fit should be treated as a sensitivity-relevant working representation rather than as a definitive biological estimate. The HR sensitivity analyses in Section~\ref{sec:p2-pw-robustness} therefore evaluate the consequences of reference-curve misspecification explicitly.

The \textsf{WCRBayesDesign} package requires two user-supplied
functions when the null survival model is non-standard: the survival
function $S_0(t)$ itself and its inverse $S_0^{-1}(u)$.  The survival
function is needed to transform the observed follow-up time
$Y_i$ into the information scale $W_i = -\log S_0(Y_i)$ used by the Bayesian
posterior update (Section~3); the inverse is needed to generate
event times through the probability-integral transform
$T = S_0^{-1}(U)$, $U \sim \mathrm{Uniform}(0,1)$.
The package receives these functions through the \texttt{S0\_fun} and
\texttt{S0\_inv\_fun} arguments.  The following code defines the two
functions for the fitted piecewise Weibull model and records the
clinical inputs used throughout the worked example.

\begin{lstlisting}
S0_pw_weibull <- function(t, k1 = 1.888498, lambda1 = 0.028189,
                          k2 = 0.07856898, lambda2 = 0.005901,
                          t_break = 24, ...) {
  H_break <- (lambda1 * t_break)^k1
  ifelse(t <= t_break,
         exp(-(lambda1 * t)^k1),
         exp(-(H_break + (lambda2 * t)^k2 - (lambda2 * t_break)^k2)))
}

S0_pw_weibull_inv <- function(surv_prob, k1 = 1.888498, lambda1 = 0.028189,
                               k2 = 0.07856898, lambda2 = 0.005901,
                               t_break = 24) {
  surv_prob <- pmin(pmax(surv_prob, .Machine$double.eps), 1)
  H_break   <- (lambda1 * t_break)^k1
  H_break2  <- (lambda2 * t_break)^k2
  ifelse(-log(surv_prob) <= H_break,
         (-log(surv_prob))^(1 / k1) / lambda1,
         (-log(surv_prob) - H_break + H_break2)^(1 / k2) / lambda2)
}

S0_par_pw <- list(k1 = 1.888498, lambda1 = 0.028189,
                  k2 = 0.07856898, lambda2 = 0.005901,
                  t_break = 24)

theta_L   <- S0_pw_weibull(24)    # 0.6200
theta_alt <- 0.80

tau       <- 24
alpha_tgt <- 0.05;  beta_tgt <- 0.20
rate      <- 5 / 12
a0 <- 0.01;  b0 <- 0.01
X_grid <- c(0, seq(tau / 8, tau, length.out = 8))
\end{lstlisting}

\noindent
The inverse function is used by the optimization and
operating-characteristic routines to generate event times via
inverse-CDF sampling.

\subsection{Landmark Walkthrough: Estimand \texorpdfstring{$S(\tau)$}{S(tau)}}
\label{sec:landmark}

We first consider the landmark estimand $S(\tau)$, the 24-month
survival probability, implemented via \texttt{method = "landmark"}.
Under the transformed exponential scale, each patient's follow-up
is mapped to $W_i = -\log S_0(Y_i)$, where $Y_i$ denotes the observed
event or censoring time available at the analysis.  Interim and final
decisions are based on the posterior tail probability
$\Pr(S(\tau) < \theta_L \mid \text{data})$.

\subsubsection{Step 1 --- Search \texorpdfstring{$(N_W, X, p_{\IA})$}{(Nw, X, pIA)}
  with \texttt{find\_Nw\_pIA()}}

The search begins by screening all $(N_W, X)$ pairs against two
futility calibration targets: the early stopping probability under
$H_0$ must be at least 0.40, and the
early stopping probability under $H_1$ must not exceed 0.10.  For each admissible pair, the function
identifies the smallest $p_{\IA}$ satisfying both constraints; pairs
with no feasible cutoff are discarded.  The following call executes
this search over the $\mathcal{X}$ grid defined above:

\begin{lstlisting}
set.seed(2025)
Nw_res <- find_Nw_pIA(
  tau = tau, theta_L = theta_L, theta_alt = theta_alt,
  alpha_target = alpha_tgt, beta_target = beta_tgt,
  rate = rate, X_grid = X_grid, num_grid = 20,
  p_H0 = 0.40, p_H1 = 0.10, nsim = 10000,
  S0_dist = "custom", S0_fun = S0_pw_weibull, S0_par = S0_par_pw,
  a0 = a0, b0 = b0, method = "landmark",
  seed = 2025, drop_na = TRUE
)
head(Nw_res, 8)
\end{lstlisting}

\emph{Output (147 feasible triples; first 8 shown):}
\begin{lstlisting}[language={},keywordstyle={},basicstyle=\ttfamily\small]
Feasible triples found: 147

   Nw  X  pIA    piF_H0     piF_H1
    6 24 0.41 0.4111     0.0972
    9 12 0.49 0.4192     0.0981
    9 15 0.46 0.4444     0.0989
    9 18 0.43 0.4511     0.0962
    9 21 0.39 0.4785     0.0973
    9 24 0.37 0.4991     0.0993
   11  6 0.49 0.4078     0.0985
   11  9 0.46 0.4386     0.0972
\end{lstlisting}

A total of 147 triples satisfy both constraints.  The feasible set
shows the expected exchange between cohort size and maturity: small
locked cohorts are viable only when paired with longer follow-up
windows, whereas short-window designs require more patients in the
interim cohort.  Longer windows provide more event information per
patient and therefore make the futility rule more sensitive.

\subsubsection{Step 2 --- Optimize \texorpdfstring{$(N_T, p_{\FA})$}{(NT, pFA)}
  and select a design}

For each admissible interim rule from Step~1, the package selects
$N_T$ and $p_{\FA}$ subject to $\alpha \leq 0.05$ and power $\geq 0.80$
and screens the resulting designs under three optimality
criteria: ESS under $H_0$ (\texttt{optimize = "ESS"}), aggregate
follow-up under $H_0$ (\texttt{optimize = "followup"}), and
maximum enrollment (\texttt{optimize = "minimax"}). In the RMS-2021
setting we treat ESS as the primary criterion because patient
availability is the binding operational constraint, but the three
criteria can disagree --- a low-ESS design may require a larger
$N_T$, a low-$N_T$ design may need a longer waiting window --- and
we report all three for completeness. The code below runs the
optimization and extracts the three benchmark designs:

\begin{lstlisting}
design_lm <- two_stage_optimize_design(
  NwX_pIA_results = Nw_res,
  rate = rate, tau = tau,
  theta_L = theta_L, theta_alt = theta_alt,
  a0 = a0, b0 = b0,
  alpha_target = alpha_tgt, beta_target = beta_tgt,
  nsim = 10000,
  S0_dist = "custom", S0_fun = S0_pw_weibull, S0_par = S0_par_pw,
  S0_inv_fun = S0_pw_weibull_inv,
  method = "landmark",
  seed = 2025, verbose = FALSE, max_iter = 10, ncores = 15
)

lm_all   <- design_lm$all
ess_idx  <- which.min(lm_all$ESS_H0)
fol_idx  <- which.min(lm_all$E_follow_H0)
mmax_idx <- which.min(lm_all$N)

lm_best   <- lm_all[ess_idx,  ]
lm_follow <- lm_all[fol_idx,  ]
lm_mmax   <- lm_all[mmax_idx, ]
\end{lstlisting}

Table~\ref{tab:lm-candidates} presents the three benchmark designs.
Under the piecewise Weibull baseline for this RMS setting, all three
criteria select the same configuration: $N_W = 15$, $X = 18$,
$p_{\IA} = 0.34$, $p_{\FA} = 0.0483$, $N_T = 41$.  This coincidence
should not be expected in general; the HR analysis in
Section~\ref{sec:hr} provides a counterexample in which the three
criteria yield distinct designs.

\begin{table}[ht]
\centering
\caption{Landmark design candidates under the piecewise Weibull working
  model.  All designs meet $\alpha \leq 0.05$ and power $\geq 0.80$.}
\label{tab:lm-candidates}
\resizebox{\textwidth}{!}{%
\begin{tabular}{lccccccccc}
\toprule
Objective & $N_W$ & $X$ & $p_{\IA}$ & $p_{\FA}$ & $N_T$
  & Type I error & Power & $\mathrm{ESS}(H_0)$
  & $E[\mathrm{follow} \mid H_0]$ \\
\midrule
Min $\mathrm{ESS}(H_0)$
  & 15 & 18 & 0.34 & 0.0483 & 41 & 0.0487 & 0.8038 & 29.89 & 554.1 mo \\
Min $E[\mathrm{follow} \mid H_0]$
  & 15 & 18 & 0.34 & 0.0483 & 41 & 0.0487 & 0.8038 & 29.89 & 554.1 mo \\
Min $N_T$
  & 15 & 18 & 0.34 & 0.0483 & 41 & 0.0487 & 0.8038 & 29.89 & 554.1 mo \\
\bottomrule
\end{tabular}}
\end{table}

When the three criteria diverge, the primary design should be chosen
according to the trial's most binding operational constraint.  Because
the three criteria coincide here, we carry this design forward to the
verification step.

\subsubsection{Step 3 --- Verify operating characteristics with
  \texttt{oc\_two\_stage()}}

We verify the selected design by rerunning the operating-characteristic
simulation with 10{,}000 replicates.  Each admissible design carries a
\texttt{seed\_used} field assigned during Step~2; passing this seed to
\texttt{oc\_two\_stage()} reproduces the design-step operating
characteristics exactly.  The same convention applies to all OC
evaluations in this manuscript.  The verification call is shown below:

\begin{lstlisting}
oc_lm <- oc_two_stage(
  N = lm_best$N, Nw = lm_best$Nw, X = lm_best$X,
  pIA = lm_best$pIA, pF = lm_best$pF,
  rate = 5/12, tau = 24,
  theta_L = theta_L, theta_alt = theta_alt,
  a0 = 0.01, b0 = 0.01,
  S0_dist = "custom", S0_fun = S0_pw_weibull, S0_par = S0_par_pw,
  S0_inv_fun = S0_pw_weibull_inv,
  method = "landmark",
  seed = lm_best$seed_used, nsim = 10000
)
\end{lstlisting}

\emph{OC summary:}
\begin{center}
\small
\begin{tabular}{llllllll}
\toprule
Type I error & Power & $\pi_F^{(0)}$ & $\pi_F^{(1)}$
  & $\mathrm{ESS}(H_0)$ & $E[\mathrm{follow} \mid H_0]$
  & $E[T_{\IA}](H_0)$ & $E[T_{\FA}](H_0)$ \\
\midrule
0.0487 & 0.8038 & 0.6024 & 0.0984 & 29.89 & 554.1 mo & 51.5 mo & 120.0 mo \\
\bottomrule
\end{tabular}
\end{center}

The one-sided type~I error rate is 0.049 and power is 0.804, both
meeting the nominal targets.  Under $H_0$, the interim stops
approximately 60\% of trials ($\pi_F^{(0)} = 0.60$), reducing the
expected sample size to 29.9 patients from a maximum enrollment of
41.  Under $H_1$, the probability of premature stopping is 0.098,
which preserves the chance of continuing when the treatment effect is
clinically meaningful.  The expected calendar time to the interim is
51.5 months, early enough for a pediatric oncology data monitoring
committee to act on the futility assessment.

\subsubsection{Step 4 --- Trial conduct with \texttt{conduct()}}
\label{sec:conduct-lm}

Given patient-level data, \texttt{conduct()} applies the prespecified
Bayesian decision rule and returns the posterior summary on the
survival-probability scale (\texttt{method = "landmark"}).  The
function expects a data frame with columns \texttt{Time} (follow-up in
months, capped at $\tau$) and \texttt{Status} (1\,=\,event,
0\,=\,censored).

To illustrate, we simulate a single trial under $H_1$ and apply the
design from Table~\ref{tab:lm-candidates}
($N_W = 15$, $X = 18$, $N_T = 41$).
Accrual times follow an exponential inter-arrival process at rate
$r = 5/12$; event times are drawn via inverse-CDF sampling from
$S_0(t)$ under $H_1$.  At the interim calendar time
$T_{\IA} = a_{(N_W)} + X = 38.88$ months, each locked patient's
follow-up is capped at $\tau = 24$ months and the event indicator
is set accordingly.  Table~\ref{tab:ia-data-lm} displays the locked
cohort; full R code for data generation and the \texttt{conduct()}
calls is given in Web Appendix~A.

\begin{table}[ht]
\centering
\caption{Simulated locked cohort ($N_W = 15$) at the interim
  analysis ($T_{\IA} = 38.88$ months) under $H_1$.
  $a_i$: enrollment time (months); Time: observed follow-up
  (capped at $\tau = 24$); Status: event indicator.}
\label{tab:ia-data-lm}
\small
\begin{tabular}{rrrr}
\toprule
$i$ & $a_i$ & Time & Status \\
\midrule
 1 &  2.02 & 24.00 & 0 \\
 2 &  3.41 & 24.00 & 0 \\
 3 &  6.60 & 24.00 & 0 \\
 4 &  6.67 & 24.00 & 0 \\
 5 &  6.81 & 24.00 & 0 \\
 6 &  7.57 & 24.00 & 0 \\
 7 &  8.32 & 24.00 & 0 \\
 8 &  8.67 & 24.00 & 0 \\
 9 & 15.21 & 23.67 & 0 \\
10 & 15.28 & 23.60 & 0 \\
11 & 17.70 & 21.18 & 0 \\
12 & 18.85 & 20.03 & 0 \\
13 & 19.52 & 19.36 & 0 \\
14 & 20.43 &  5.78 & 1 \\
15 & 20.88 & 16.75 & 1 \\
\bottomrule
\end{tabular}
\end{table}

The interim analysis applies \texttt{conduct()} to this locked
cohort (see Web Appendix~A for the call).

\emph{Interim result:}
\begin{lstlisting}[language={},keywordstyle={},basicstyle=\ttfamily\small]
$Current.sample.size
[1] 15

$Posterior.shape
[1] 2.01

$Posterior.rate
[1] 6.0741

$`Posterior tail probability (Pr >= threshold)`
[1] 0.0165

$`Posterior mean of survival rate`
[1] 0.8537

$`95% credible interval of survival rate`
[1] 0.6441 0.9809

$Decision
[1] "Go (Continue)"
\end{lstlisting}

Of the 15 locked patients, 2 experience events within $\tau = 24$
months.  The posterior mean 2-year survival is $\hat\theta = 0.85$,
with 95\% credible interval $[0.64,\, 0.98]$.  The posterior tail
probability is 0.017, well below the futility cutoff
$p_{\IA} = 0.34$; the trial therefore continues to full enrollment.

After all $N_T = 41$ patients complete $\tau = 24$ months of
follow-up, the final analysis proceeds.

\emph{Final result:}
\begin{lstlisting}[language={},keywordstyle={},basicstyle=\ttfamily\small]
$Current.sample.size
[1] 41

$Posterior.shape
[1] 8.01

$Posterior.rate
[1] 17.1184

$`Posterior tail probability (Pr >= threshold)`
[1] 0.0051

$`Posterior mean of survival rate`
[1] 0.7996

$`95% credible interval of survival rate`
[1] 0.6682 0.9079

$Decision
[1] "Efficacious"
\end{lstlisting}

With all 41 patients contributing, 8 experience events within
$\tau = 24$ months.  The posterior mean 2-year survival is
$\hat\theta = 0.80$ (95\% credible interval $[0.67,\, 0.91]$), and
the posterior tail probability is 0.005.  Since this probability is
below the efficacy cutoff $p_{\FA} = 0.0483$, the trial concludes
``Efficacious.''

\subsection{Hazard-Ratio Walkthrough: Estimand \texorpdfstring{$\delta$}{delta}}
\label{sec:hr}

We now consider the proportional-hazards estimand
$\delta = \log S(\tau) / \log S_0(\tau)$, accessed via
\texttt{method = "hr"}.  Under this parameterization,
$S_1(t) = S_0(t)^\delta$ for all $t$, and $\delta < 1$ indicates a
treatment benefit.  The design alternative corresponds to
$\delta_{\mathrm{alt}} = \log(0.80)/\log(0.62) \approx 0.4667$,
yielding $S_1(24) = 0.80$.  The three-step workflow parallels
Section~\ref{sec:landmark}; full R code is provided in
Web Appendix~B.

\subsubsection{Step 1 --- Feasible \texorpdfstring{$(N_W, X, p_{\IA})$}{(Nw, X, pIA)}
  triples}

Applying the same calibration targets (early stopping probability
$\geq 0.40$ under $H_0$ and $\leq 0.10$ under $H_1$), the HR-mode
search produces 148 feasible interim rules.  The first six are shown below.

\begin{center}
\small
\begin{tabular}{rrrrr}
\toprule
$N_W$ & $X$ & $p_{\IA}$ & $\pi_F^{(0)}$ & $\pi_F^{(1)}$ \\
\midrule
6 & 21 & 0.43 & 0.4036 & 0.0959 \\
6 & 24 & 0.48 & 0.4226 & 0.0998 \\
9 & 12 & 0.48 & 0.4271 & 0.0997 \\
9 & 15 & 0.46 & 0.4438 & 0.0963 \\
9 & 18 & 0.42 & 0.4635 & 0.0992 \\
9 & 21 & 0.38 & 0.4929 & 0.0989 \\
\bottomrule
\end{tabular}
\end{center}

As in the landmark analysis, increasing $X$ raises the null
futility-stop rate at the cost of a later interim calendar time.

\subsubsection{Step 2 --- Candidate designs}

Optimization with \texttt{method = "hr"} produces 148 admissible
designs.  Table~\ref{tab:hr-candidates} presents the three benchmark
candidates.  In contrast to the landmark case, the three optimality
criteria now select distinct designs.  The HR likelihood extracts
information from the full event-time distribution rather than reducing
each patient to a binary indicator at $\tau$, so the trade-offs among
$N_W$, $X$, and $N_T$ are less uniform.

\begin{table}[ht]
\centering
\caption{HR-mode design candidates under the piecewise Weibull working
  model ($\theta_L = 0.62$, $\theta_{\mathrm{alt}} = 0.80$, $r = 5/12$,
  nsim\,=\,10\,000, num\_grid\,=\,20).}
\label{tab:hr-candidates}
\resizebox{\textwidth}{!}{%
\begin{tabular}{lcccccccccc}
\toprule
Objective & $N_W$ & $X$ & $p_{\IA}$ & $p_{\FA}$ & $N_T$ & Type I error
  & Power & $\mathrm{ESS}(H_0)$ & $E[\mathrm{follow} \mid H_0]$
  & $E[T_{\IA}]$ \\
\midrule
Min $\mathrm{ESS}(H_0)$
  & 19 & 3 & 0.34 & 0.0406 & 39 & 0.0455 & 0.8017 & 28.02 & 1029.2 mo & 46.1 mo \\
Min $E[\mathrm{follow} \mid H_0]$
  & 17 & 15 & 0.29 & 0.0483 & 37 & 0.0479 & 0.8058 & 28.12 & 961.8 mo & 53.3 mo \\
Min $N_T$
  & 25 & 9 & 0.21 & 0.0406 & 36 & 0.0500 & 0.8022 & 30.65 & 997.5 mo & 66.6 mo \\
\bottomrule
\end{tabular}}
\end{table}

The ESS-optimal design ($N_W = 19$, $X = 3$) triggers the interim
early and achieves the lowest expected enrollment under $H_0$ (28.0
patients) but accumulates the highest aggregate follow-up (1029
person-months).  The follow-up-optimal design ($N_W = 17$, $X = 15$)
delays the interim to allow more maturity per patient, reducing
aggregate follow-up to 962 person-months at a modest ESS increase.
The minimax design ($N_W = 25$, $X = 9$) caps total enrollment at
$N_T = 36$ but carries the highest ESS (30.7), because fewer null
trials trigger early stopping.  For RMS2021, where the patient pool is
limited, the minimax design may be the most operationally defensible;
if rapid termination of futile regimens is the priority, the ESS
design is preferable.

We verify all three candidate designs via independent
\texttt{oc\_two\_stage()} runs.

\subsubsection{Step 3 --- OC verification}

Unlike the landmark case, where the three optimality criteria selected
a single design, the HR-mode candidates differ in their
$(N_W, X, N_T)$ configurations.  We therefore run
\texttt{oc\_two\_stage()} for each candidate
($\text{nsim} = 10{,}000$) and report the results in
Table~\ref{tab:hr-oc}.

\begin{table}[ht]
\centering
\caption{OC verification of the three HR-mode candidate designs
  ($\text{nsim} = 10{,}000$, each using its design-step \texttt{seed\_used}).}
\label{tab:hr-oc}
\resizebox{\textwidth}{!}{%
\begin{tabular}{lcccccccc}
\toprule
Design & Type I error & Power & $\pi_F^{(0)}$ & $\pi_F^{(1)}$
  & $\mathrm{ESS}(H_0)$ & $E[\mathrm{follow} \mid H_0]$
  & $E[T_{\IA}](H_0)$ & $E[T_{\FA}](H_0)$ \\
\midrule
Min $\mathrm{ESS}(H_0)$
  & 0.0455 & 0.8017 & 0.5854 & 0.0959 & 28.02 & 1029.2 mo & 46.1 mo & 114.6 mo \\
Min $E[\mathrm{follow} \mid H_0]$
  & 0.0479 & 0.8058 & 0.6449 & 0.0966 & 28.12 & 961.8 mo & 53.3 mo & 110.0 mo \\
Min $N_T$
  & 0.0500 & 0.8022 & 0.7377 & 0.0958 & 30.65 & 997.5 mo & 66.6 mo & 108.2 mo \\
\bottomrule
\end{tabular}}
\end{table}

All three HR candidates meet the nominal targets ($\alpha \le 0.05$
and power $\ge 0.80$).  The substantive contrasts
are clear: the follow-up-optimal
design gives the lowest aggregate follow-up (962 person-months), the
ESS-optimal design gives the lowest expected enrollment (28.0
patients), and the minimax design gives the smallest maximum sample
size ($N_T = 36$) with the highest null futility-stop rate
($\pi_F^{(0)} = 0.74$).

\subsubsection{Step 4 --- Trial conduct with \texttt{conduct()}}
\label{sec:conduct-hr}

The HR analogue uses the ESS-optimal design ($N_W = 19$, $X = 3$,
$N_T = 39$) and \texttt{method = "hr"}, so that \texttt{conduct()}
reports the posterior on the $\delta$ scale rather than the
survival-probability scale.  The procedural difference from the
landmark version lies in censoring: at the interim, follow-up is not
truncated at $\tau$, so early enrollees may contribute substantially
more than 24 months of information; at the final analysis, the study
ends at $t_{\mathrm{final}} = a_{(N)} + \tau$ and each patient's
follow-up is $f_i = t_{\mathrm{final}} - a_i$.

We reuse the same $H_1$-generation process described in
Section~\ref{sec:conduct-lm}, now with the HR design parameters
($N_W = 19$, $X = 3$, $N_T = 39$) and without truncating follow-up
at $\tau$.  The \texttt{conduct()} calls are given in
Web Appendix~A.

\emph{Interim result:}
\begin{lstlisting}[language={},keywordstyle={},basicstyle=\ttfamily\small]
$Current.sample.size
[1] 19

$Posterior.shape
[1] 0.01

$Posterior.rate
[1] 5.2254

$`Posterior tail probability (Pr >= threshold)`
[1] 0.0000

$`Posterior mean of delta`
[1] 0.0019

$`95% credible interval of delta`
[1] 0.0000 0.0090

$Decision
[1] "Go (Continue)"
\end{lstlisting}

With $X = 3$ months of additional follow-up, none of the 19 locked
patients experience events by the interim calendar time; the 19
censored patients contribute follow-up periods of varying lengths,
yielding a total transformed time of 5.23.  The posterior is dominated
by the prior and the accumulated censoring: $\hat\delta = 0.002$ with
95\% credible interval $[0.00,\, 0.01]$.  Because
$\Pr(\delta \geq \delta_L \mid \text{data}) \approx 0$ falls well
below $p_{\IA} = 0.34$, the trial continues.

\emph{Final result:}
\begin{lstlisting}[language={},keywordstyle={},basicstyle=\ttfamily\small]
$Current.sample.size
[1] 39

$Posterior.shape
[1] 8.01

$Posterior.rate
[1] 19.3021

$`Posterior tail probability (Pr >= threshold)`
[1] 0.0013

$`Posterior mean of delta`
[1] 0.4150

$`95% credible interval of delta`
[1] 0.1793 0.7479

$Decision
[1] "Efficacious"
\end{lstlisting}

The full trial enrolls 39 patients, of whom 8 experience events by
the final calendar time.  The posterior mean
$\hat\delta = 0.42$ with 95\% credible interval $[0.18,\, 0.75]$
indicates a clear treatment benefit; the interval excludes the null
value $\delta_L = 1.0$.  The tail probability is 0.001, which falls
below $p_{\FA} = 0.0406$, so the trial concludes ``Efficacious.''

\subsubsection{Landmark vs.\ hazard-ratio comparison}

Table~\ref{tab:comparison} compares the operating characteristics of
the selected landmark and HR designs under the same baseline model and
accrual assumptions. These rows compare optimized design recommendations under two estimands; they should not be interpreted as isolating the likelihood effect while holding $(N_W,X,N_T)$ fixed.

\begin{table}[ht]
\centering
\caption{Operating characteristics of the selected landmark and HR
  designs under the piecewise Weibull model
  ($r = 5/12$, $\theta_{\mathrm{alt}} = 0.80$,
  nsim\,=\,10\,000).}
\label{tab:comparison}
\small
\begin{tabular}{lcccccc}
\toprule
Design & $N_T$ & Type I error & Power & $\pi_F^{(0)}$
  & $\mathrm{ESS}(H_0)$ & $E[\mathrm{follow} \mid H_0]$ \\
\midrule
Landmark & 41 & 0.0487 & 0.804 & 0.602 & 29.89 & 554.1 mo \\
Hazard ratio & 39 & 0.0455 & 0.802 & 0.585 & 28.02 & 1029.2 mo \\
\bottomrule
\end{tabular}
\end{table}

The HR design achieves a smaller maximum sample size ($N_T = 39$
vs.\ 41) because it extracts information from the full event-time
distribution, but this efficiency gain is offset by a substantially
higher follow-up burden: 1029 person-months under $H_0$ compared with
554 for the landmark design.  In a multi-site pediatric trial, the
additional person-months translate into more clinic visits, more data
collection, and greater site coordination.  The two designs have
comparable futility-stop rates ($\pi_F^{(0)} = 0.59$ vs.\ 0.60)
and similar expected interim calendar times ($E[T_{\IA}] \approx
46$--52 months), so neither early-stopping performance nor monitoring
schedule differentiates the two.

The estimand should be chosen from the clinical question rather than
from sample size alone.  RMS2021 is designed around 2-year EFS, and
the ARST0531 historical reference is defined at the same landmark;
the landmark estimand is therefore the natural match.  It has a
direct clinical interpretation and a simpler follow-up protocol, since
each patient's contribution is capped at $\tau$ months.  The HR
estimand can be more statistically efficient when proportional hazards
is credible, but it requires heavier follow-up and a more complex
monitoring structure.  For this application, we recommend the
landmark design.  The HR version is better suited to trials with a
tighter enrollment cap, credible prior support for proportional
hazards, and enough follow-up capacity to support the additional
information burden.

\subsubsection{Prior-sensitivity check}

The preceding analyses use the weak Gamma$(0.01,0.01)$ prior on the
transformed hazard parameter.  To check that the design conclusions are
not driven by this near-flat prior, we re-evaluate the selected
landmark and HR designs under two heavier but still weakly informative
priors, Gamma$(0.1,0.1)$ and Gamma$(1,1)$, holding the stopping rules
fixed.  The simulations use the same piecewise Weibull working model,
the same calibrated thresholds, and $10{,}000$ replicates with the
design-anchored seeds used in the verification step.

\begin{table}[ht]
\centering
\caption{Prior-sensitivity check.  Design parameters fixed as in
  Tables~\ref{tab:lm-candidates}--\ref{tab:hr-candidates};
  only the Gamma$(a_0,b_0)$ prior is varied ($10{,}000$ replicates).}
\label{tab:prior-sensitivity}
\small
\begin{tabular}{llcccccc}
\toprule
Estimand & Prior $(a_0,b_0)$ & Type I error & Power
  & $\pi_F^{(0)}$ & $\pi_F^{(1)}$ & $\mathrm{ESS}(H_0)$
  & $E[\mathrm{follow}\mid H_0]$ \\
\midrule
Landmark & $(0.01,0.01)$ & 0.0487 & 0.8039 & 0.6024 & 0.0984 & 29.89 & 554.1 \\
Landmark & $(0.10,0.10)$ & 0.0482 & 0.8019 & 0.6039 & 0.0996 & 29.86 & 553.4 \\
Landmark & $(1.00,1.00)$ & 0.0435 & 0.7863 & 0.6202 & 0.1078 & 29.56 & 545.6 \\
HR       & $(0.01,0.01)$ & 0.0454 & 0.8014 & 0.5855 & 0.0960 & 28.02 & 1029.1 \\
HR       & $(0.10,0.10)$ & 0.0450 & 0.8001 & 0.5865 & 0.0965 & 28.00 & 1027.6 \\
HR       & $(1.00,1.00)$ & 0.0400 & 0.7846 & 0.6022 & 0.1047 & 27.70 & 1004.3 \\
\bottomrule
\end{tabular}
\end{table}

The Gamma$(0.1,0.1)$ prior leaves the operating characteristics almost
unchanged for both estimands: type I error, power, and ESS differ from
the baseline prior by less than 0.003 in rejection probability and
less than 0.03 patients in ESS.  The stronger Gamma$(1,1)$ prior is
more conservative, reducing power to 0.786 for the landmark design and
0.785 for the HR design, while also lowering type I error.  This
pattern is consistent with the prior placing more weight near the
null-scale hazard ratio.  The practical conclusion is that the
near-flat Gamma$(0.01,0.01)$ and Gamma$(0.1,0.1)$ priors give the same
design recommendation, whereas Gamma$(1,1)$ is conservative enough to
warrant recalibration if chosen for a protocol.


\section{Simulation Study}
\label{sec:simulation}

\subsection{From a Worked Example to a Systematic Sensitivity Study}

Section~4 calibrates and verifies one WCR design for RMS2021: the
reference survival $S_0$ is the piecewise Weibull fit to the ARST0531
VAC-arm Kaplan--Meier curve, the design's objective is to optimize the ESS under
$H_0$, and the primary estimand is the landmark survival at
$\tau = 24$ months.  That worked example is intentionally specific.
It does not show how sensitive the WCR rule is to the assumed baseline
hazard, the realized accrual rate, and the true survival probability.
This section evaluates those dependencies systematically. For clarity, the Section~5 sensitivity grids use a single-Weibull reference family anchored at $S_0(24)=0.62$ rather than the piecewise-Weibull RMS2021 fit used in Section~4; consequently, optimized $(N_W,X,N_T)$ values in this section are sensitivity-study recommendations under their own working models, not contradictions of the Section~4 RMS2021 worked example.

The first sensitivity analysis varies the shape of the baseline hazard.
The Bayesian calibration depends on $S_0$ through the transformed
information times $W_i = -\log S_0(Y_i)$, where $Y_i$ denotes the
observed event or censoring time available for patient $i$ at the
analysis. In practice, however, the planner may have limited confidence
regarding the correct parametric form of $S_0$ prior to trial initiation.
We therefore replace the piecewise Weibull fit by the one-parameter
Weibull family
\[
S_0(t)=\exp[-(\lambda t)^k],
\]
with $\lambda$ chosen such that $S_0(24)=0.62$.

The shape parameter
$k \in \{0.8,1.0,1.2\}$ spans decreasing, constant, and increasing
hazard regimes. The exponential case $k=1$ serves as the neutral
reference scenario, whereas $k=0.8$ corresponds to a decreasing hazard with front-loaded event accumulation and $k=1.2$ corresponds to an increasing hazard with back-loaded event accumulation. 

From an interim monitoring perspective, varying $k$ changes the temporal
rate at which inferential information emerges over calendar time,
thereby creating settings with relatively early versus delayed
information maturity at interim analysis. These scenarios therefore
assess not only robustness to baseline hazard misspecification, but also
the stability of the proposed framework under differing temporal
information structures and non-idealized hazard dynamics.

The second sensitivity analysis fixes the design calibrated under
$k = 1$ and re-evaluates its operating characteristics over a
$3 \times 5$ grid of annual accrual rates and true survival
probabilities. The proposed design is calibrated under assumed accrual
and event-generation dynamics; however, both quantities are typically
subject to substantial uncertainty in practice. In particular, accrual
rates directly influence the relationship between calendar time and
follow-up maturity, whereas the underlying survival distribution governs
the temporal rate at which inferential information accumulates.

This sensitivity analysis therefore evaluates the robustness of the
proposed information--time alignment framework under joint operational
and clinical misspecification. Varying the accrual rate alters the
maturity structure of the cohort available at interim analysis and may
substantially change the amount of effective follow-up observed by a
given calendar time. Simultaneously, varying the true survival
probability modifies the event-generation process and consequently the
rate of information emergence over time. Together, these scenarios
assess whether the calibrated design maintains stable operating
characteristics when the realized calendar-time trajectory of
information maturity deviates from the assumptions used during the
design stage.

The analysis also expands the design space beyond the single
Section~4 recommendation. We evaluate both estimands and all three
optimization criteria implemented in the package: ESS, total
follow-up, and minimax-$N_T$. More importantly, the analysis explicitly
characterizes the ethical and operational burden induced by interim
timing mechanisms themselves. In long-horizon and sparse-event settings,
prolonged waiting periods prior to interim decisions may expose
additional patients to potentially ineffective therapy before futility
can be declared. Under conventional monitoring paradigms, this burden is
often treated as incidental or unavoidable. In contrast, the proposed
framework treats such exposure as an explicit and measurable operating
characteristic of the monitoring design.

To quantify this phenomenon, we introduce two ethical-burden metrics
describing the operational profile of the WCR framework. Let
$\mathcal{I} = \{i : a_i \in (a_{(N_W)},\, a_{(N_W)} + X]\}$ index
the patients enrolled during the waiting window and let
$T_{\IA} = a_{(N_W)} + X$ denote the calendar time of the interim
analysis. Define
\begin{equation}
  B_W \;=\; \mathbb{E}\bigl[\,|\mathcal{I}|\,\bigr],
  \qquad
  \bar F_W \;=\;
  \mathbb{E}\!\left[
    \frac{1}{|\mathcal{I}|}
    \sum_{i \in \mathcal{I}}
    \min\bigl(T_i,\; T_{\IA} - a_i,\; \tau\bigr)
    \,\Bigm|\, |\mathcal{I}| > 0
  \right],
  \label{eq:ethical-burden}
\end{equation}
where the per-patient quantity inside the expectation represents the
observed follow-up accrued prior to the interim decision, censored at
the event time $T_i$ and, under landmark estimands, at the clinically
relevant horizon $\tau$. The truncation at $\tau$ reflects the fact
that follow-up accrued beyond the landmark horizon does not contribute
additional inferential information under the restricted estimand
framework. In contrast, under proportional hazards--based estimands
with unrestricted follow-up usage, the truncation at $\tau$ may be
removed, since follow-up beyond $\tau$ continues to contribute to the
risk-set structure and inferential information. Consequently, the
proposed ethical-burden metrics are themselves estimand-dependent and
should be interpreted relative to the effective information horizon
associated with the target estimand. By construction,
$\bar F_W \le X$ under the restricted framework. When
$|\mathcal{I}|=0$, the realized enrollment burden and aggregate follow-up
burden are both defined as zero; the conditional mean $\bar F_W$ is reported
only for simulations with at least one waiting-window enrollee.

Together, $B_W$ and $\bar F_W$ quantify two complementary dimensions of
decision-lag burden: the expected number of additional patients exposed
during the waiting window and the average duration of interim-period
follow-up accrued by these patients prior to decision-making. These
quantities permit direct comparison of monitoring paradigms with respect
to ethically relevant operational exposure, thereby extending evaluation
beyond conventional metrics such as type I error, power, or expected
sample size alone.

Importantly, although the current work treats ethical burden primarily
as a descriptive operational characteristic, the proposed quantities
naturally admit incorporation into constrained or utility-based design
optimization frameworks. For example, future calibration strategies may
jointly optimize inferential maturity, calendar efficiency, and ethical
exposure under user-specified constraints. Such extensions may provide a
more principled framework for balancing statistical information against
patient-level exposure in modern adaptive monitoring settings. After the
sensitivity analyses are complete, Section~\ref{sec:comparison}
benchmarks WCR against existing single-arm monitoring designs using both
traditional operating characteristics and the proposed ethical-burden
metrics.
\subsection{Sensitivity to the Weibull Shape}
\label{sec:1A}

Table~\ref{tab:setup} lists the inputs held fixed across the
sensitivity scenarios.  They match the RMS2021 inputs in Section~4,
except that the piecewise Weibull reference is replaced by a
single-piece Weibull reference.  For
$k \in \{0.8, 1.0, 1.2\}$, the constraint $S_0(24) = 0.62$ gives
$\lambda = 0.01656$, $0.01992$, and $0.02253$, respectively.  Within
each $(k, \text{estimand})$ cell, we run the full calibration
pipeline and then verify the ESS-optimal, follow-up-optimal, and
minimax-$N_T$ designs by independent simulation.
Tables~\ref{tab:1A-landmark} and~\ref{tab:1A-hr} report the eighteen
resulting designs.

\begin{table}[ht]
\centering
\caption{Common simulation setup for the WCR sensitivity analysis.}
\label{tab:setup}
\begin{tabular}{ll}
\toprule
Parameter & Value \\
\midrule
Landmark horizon $\tau$ & 24 months \\
Null survival $\theta_L$ & 0.62 \\
Design alternative $\theta_{\mathrm{alt}}$ & 0.80 \\
Reference $S_0(t)$ & $\exp\bigl[-(\lambda t)^k\bigr]$, with
  $\lambda$ such that $S_0(24) = 0.62$ \\
Gamma prior $(a_0, b_0)$ & $(0.01, 0.01)$ \\
Type I target $\alpha$ & 0.05 \\
Power target $1 - \beta$ & 0.80 \\
IA futility calibration & $\pi_F^{(0)} \geq 0.40$,
  $\pi_F^{(1)} \leq 0.10$ \\
Follow-up grid $\mathcal{X}$ & $\{0,3,6,9,12,15,18,21,24\}$ months \\
\bottomrule
\end{tabular}
\end{table}

\begin{table}[ht]
\centering
\caption{Weibull-shape sensitivity, landmark-estimand designs
  ($r = 5/12$, $\theta_L = 0.62$, $\theta_{\mathrm{alt}} = 0.80$,
  $10{,}000$ replicates).  All OC metrics are under $H_0$;
  $\Bw$ and $\Fw$ are the waiting-window enrollment count and
  mean per-patient follow-up at $T_{\IA}$.}
\label{tab:1A-landmark}
\footnotesize
\begin{tabular}{llcccccccccccc}
\toprule
$k$ & Criterion & $N_W$ & $X$ & $p_{\IA}$ & $p_{\FA}$ & $N_T$
  & $\alpha$ & Power & $\mathrm{ESS}$ & $E[T_{\IA}]$ & $E[T_{\FA}]$
  & $\Bw$ & $\Fw$ \\
\midrule
\multirow{3}{*}{0.8}
  & ESS         & 17 &  6 & 0.330 & 0.0406 & 43 & 0.0479 & 0.801 & 28.95 &  44.3 & 124.4 & 2.51 & 2.56 \\
  & Follow-up   & 17 &  6 & 0.330 & 0.0406 & 43 & 0.0479 & 0.801 & 28.95 &  44.3 & 124.4 & 2.51 & 2.56 \\
  & Minimax     & 15 & 18 & 0.330 & 0.0483 & 41 & 0.0493 & 0.807 & 29.92 &  51.5 & 120.0 & 7.56 & 7.75 \\
\midrule
\multirow{3}{*}{1.0}
  & ESS         & 19 &  3 & 0.310 & 0.0406 & 45 & 0.0432 & 0.810 & 29.59 &  46.1 & 128.8 & 1.24 & 1.04 \\
  & Follow-up   & 23 &  0 & 0.280 & 0.0406 & 43 & 0.0463 & 0.810 & 29.86 &  52.9 & 124.5 & 0.00 & 0.00 \\
  & Minimax     & 15 & 18 & 0.330 & 0.0483 & 41 & 0.0490 & 0.806 & 29.84 &  51.5 & 119.9 & 7.56 & 7.99 \\
\midrule
\multirow{3}{*}{1.2}
  & ESS         & 15 & 18 & 0.330 & 0.0483 & 41 & 0.0488 & 0.804 & 29.75 &  51.5 & 119.9 & 7.56 & 8.18 \\
  & Follow-up   & 23 &  3 & 0.260 & 0.0406 & 43 & 0.0481 & 0.809 & 30.20 &  55.7 & 124.5 & 1.24 & 1.04 \\
  & Minimax     & 15 & 18 & 0.330 & 0.0483 & 41 & 0.0488 & 0.804 & 29.75 &  51.5 & 119.9 & 7.56 & 8.18 \\
\bottomrule
\end{tabular}
\end{table}

\begin{table}[ht]
\centering
\caption{Weibull-shape sensitivity, hazard-ratio--estimand designs. Conventions as in
  Table~\ref{tab:1A-landmark}. The HR analysis does not truncate
  follow-up at $\tau$, so $E[T_{\FA}]$ is generally smaller than for
  the landmark designs at comparable $X$.}
\label{tab:1A-hr}
\footnotesize
\begin{tabular}{llcccccccccccc}
\toprule
$k$ & Criterion & $N_W$ & $X$ & $p_{\IA}$ & $p_{\FA}$ & $N_T$
  & $\alpha$ & Power & $\mathrm{ESS}$ & $E[T_{\IA}]$ & $E[T_{\FA}]$
  & $\Bw$ & $\Fw$ \\
\midrule
\multirow{3}{*}{0.8}
  & ESS         & 15 &  0 & 0.400 & 0.0406 & 28 & 0.0411 & 0.811 & 21.24 &  33.7 & 88.8 & 0.00 & 0.00 \\
  & Follow-up   & 19 &  0 & 0.300 & 0.0406 & 27 & 0.0443 & 0.811 & 21.98 &  43.2 & 86.2 & 0.00 & 0.00 \\
  & Minimax     & 19 &  0 & 0.300 & 0.0406 & 27 & 0.0443 & 0.811 & 21.98 &  43.2 & 86.2 & 0.00 & 0.00 \\
\midrule
\multirow{3}{*}{1.0}
  & ESS         & 13 &  3 & 0.410 & 0.0483 & 25 & 0.0482 & 0.812 & 19.53 &  31.8 & 81.5 & 1.25 & 1.04 \\
  & Follow-up   & 17 &  3 & 0.320 & 0.0483 & 24 & 0.0489 & 0.812 & 20.41 &  41.3 & 79.1 & 1.22 & 1.03 \\
  & Minimax     & 17 &  3 & 0.320 & 0.0483 & 24 & 0.0489 & 0.812 & 20.41 &  41.3 & 79.1 & 1.22 & 1.03 \\
\midrule
\multirow{3}{*}{1.2}
  & ESS         & 13 &  3 & 0.420 & 0.0330 & 25 & 0.0336 & 0.804 & 19.71 &  31.8 & 81.5 & 1.25 & 1.04 \\
  & Follow-up   & 11 & 15 & 0.370 & 0.0406 & 23 & 0.0429 & 0.802 & 19.76 &  38.9 & 76.6 & 6.22 & 6.98 \\
  & Minimax     & 11 & 15 & 0.370 & 0.0406 & 23 & 0.0429 & 0.802 & 19.76 &  38.9 & 76.6 & 6.22 & 6.98 \\
\bottomrule
\end{tabular}
\end{table}

All eighteen designs satisfy the $\alpha \le 0.05$ and power
$\ge 0.80$ constraints under the simulation seed at which the
optimizer accepts them.  Type~I error ranges from 0.0432 to 0.0493
for the landmark designs and from 0.0336 to 0.0489 for the HR
designs; power ranges from 0.801 to 0.810 for landmark and from
0.802 to 0.812 for HR.  The HR estimand requires much smaller maximum
sample sizes than the landmark estimand ($N_T \in \{23, \dots, 28\}$
versus $\{41, 43, 45\}$), because its likelihood uses the full
event-time distribution rather than the binary survival indicator at
$\tau$.

The three optimization criteria target different aspects of trial
economics.  The ESS-optimal criterion lowers the expected sample size
under $H_0$ by stopping earlier when the regimen is ineffective, but
often does so by allowing a larger maximum $N_T$.  At $k = 1$ under
the landmark estimand, for example, the ESS design uses $N_T = 45$,
whereas the follow-up and minimax designs use 43 and 41.  The minimax
criterion reduces $N_T$ by lengthening the waiting window: every
minimax landmark design selects $X = 18$, compared with
$X \in \{3, 6\}$ for the ESS designs.  The follow-up-optimal
criterion shortens the expected final-analysis time but may enroll
more patients during the waiting window.  In the $k = 1.2$ HR cell,
the follow-up-optimal design has $\Bw = 6.22$ patients enrolled
during the waiting window, compared with $\Bw = 1.25$ for the ESS
design.

From an information--time perspective, varying the Weibull shape
parameter changes the temporal rate at which inferential information
emerges over calendar time. Increasing hazards concentrate event
accumulation later in the follow-up trajectory, whereas decreasing
hazards generate earlier information emergence. Consequently, the
effective maturity profile available at interim analysis depends not
only on the waiting window $X$, but also on how rapidly events become
observable within the locked cohort.

The optimizer responds to these differing information-emergence
dynamics by rebalancing the trade-off between locked-cohort follow-up
and continued enrollment. Under increasing hazards $(k = 1.2)$,
events accumulate later in each patient's trajectory, so early interim
datasets contain relatively immature event information. The landmark
ESS-optimal design therefore lengthens the waiting window from
$X = 3$ at $k = 1$ to $X = 18$, while simultaneously reducing the
locked-cohort size from $N_W = 19$ to $N_W = 15$, thereby allowing
additional calendar time for event information to mature within the
locked cohort prior to interim evaluation. Under decreasing hazards
$(k = 0.8)$, events occur earlier and inferential information emerges
more rapidly, allowing a substantially shorter waiting window
($X = 6$) while maintaining the required operating characteristics.

The HR estimand exhibits substantially less sensitivity to Weibull
shape variation because the unrestricted likelihood continues to use
follow-up information beyond the landmark horizon $\tau$. As a result,
the effective information horizon is broader and less dependent on the
precise timing of early event accumulation. Consequently, the
ESS-optimal HR designs maintain relatively stable waiting-window
choices ($X \in \{0,3\}$) across all three hazard shapes, with the
optimizer primarily adjusting the locked-cohort size $N_W$ rather than
substantially altering the follow-up window itself. This distinction
highlights an important conceptual feature of the proposed framework:
the interaction between interim timing and information maturity depends
not only on the event-generation process, but also on the estimand
through which follow-up information is utilized.

\subsection{Sensitivity to Accrual Rate and True Survival}
\label{sec:1B}

We now hold the design fixed at the values calibrated under $k = 1$
in the previous section and probe its operating characteristics on a
$3 \times 5$ grid spanning
$\text{annual accrual} \in \{2.5, 5, 10\}$ patients per year and
$\theta_{\mathrm{truth}} \in \{0.62, 0.70, 0.75, 0.80, 0.85\}$. The
$\theta_{\mathrm{truth}} = 0.62$ column reports the empirical type I
error, and the $\theta_{\mathrm{truth}} = 0.80$ column reports the
power against the alternative used at calibration. The main text
focuses on the ESS-optimal designs, which is the natural primary
criterion when patient resources are the binding constraint and
the criterion we adopt as primary throughout. The
follow-up-optimal and minimax designs are deferred to
Web Appendix~C. 

\subsubsection{ESS-Optimal Landmark Design}

The ESS-optimal landmark design at $k = 1$ is
$(N_W, X, p_{\IA}, p_{\FA}, N_T) = (19,\, 3,\, 0.310,\, 0.0406,\, 45)$.
Table~\ref{tab:1B-lm-ess} consolidates its operating characteristics
across the rate-by-$\theta_{\mathrm{truth}}$ grid. Panel~A reports
the rejection probability, the only metric that depends on
$\theta_{\mathrm{truth}}$ in any practical sense; the
$\theta_{\mathrm{truth}} = 0.62$ column is the empirical type I
error and the $0.80$ column is the power against the design
alternative. Panel~B reports the calendar-timing and ethical-burden
metrics, all of which are invariant to $\theta_{\mathrm{truth}}$ to
within Monte Carlo noise --- $E[T_{\IA}]$ exactly so by the WCR
trigger, $E[T_{\FA}]$ to within $0.5$~months across the
$\theta_{\mathrm{truth}}$ values, and $\Bw$, $\Fw$ exactly so
because they are functions of accrual and $X$ only.

\begin{table}[ht]
\centering
\caption{Landmark ESS design at $k = 1$,
  $(N_W, X, N_T) = (19, 3, 45)$.  Panel~A: rejection probability.
  Panel~B: timing and ethical-burden metrics at
  $\theta_{\mathrm{truth}} = 0.80$.}
\label{tab:1B-lm-ess}
\small
\begin{tabular}{lccccc}
\toprule
\multicolumn{6}{l}{\textbf{Panel A.} $\Pr(\text{reject } H_0)$} \\[2pt]
\multirow{2}{*}{Annual accrual} & \multicolumn{5}{c}{$\theta_{\mathrm{truth}}$} \\
\cmidrule(lr){2-6}
 & 0.62 & 0.70 & 0.75 & 0.80 & 0.85 \\
\midrule
2.5 / yr  & 0.0447 & 0.2743 & 0.5478 & 0.8162 & 0.9540 \\
5.0 / yr  & 0.0432 & 0.2681 & 0.5404 & 0.8101 & 0.9475 \\
10.0 / yr & 0.0407 & 0.2543 & 0.5183 & 0.7790 & 0.9290 \\
\midrule
\multicolumn{6}{l}{\textbf{Panel B.} Timing and ethical burden} \\[2pt]
Annual accrual & $E[T_{\IA}]$ & $E[T_{\FA}]$ & $\Bw$ & $\Fw$ & \\
\midrule
2.5 / yr  &  89.19 & 234.81 & 0.62 & 0.67 & \\
5.0 / yr  &  46.09 & 129.39 & 1.24 & 1.04 & \\
10.0 / yr &  24.55 &  76.79 & 2.49 & 1.34 & \\
\bottomrule
\end{tabular}
\end{table}

\subsubsection{ESS-Optimal HR Design}

The ESS-optimal hazard-ratio design at $k = 1$ is
$(N_W, X, p_{\IA}, p_{\FA}, N_T) = (13,\, 3,\, 0.410,\, 0.0483,\, 25)$.
Table~\ref{tab:1B-hr-ess} reports the analogous panel structure.

\begin{table}[ht]
\centering
\caption{HR ESS design at $k = 1$,
  $(N_W, X, N_T) = (13, 3, 25)$. Conventions as in
  Table~\ref{tab:1B-lm-ess}.}
\label{tab:1B-hr-ess}
\small
\begin{tabular}{lccccc}
\toprule
\multicolumn{6}{l}{\textbf{Panel A.} $\Pr(\text{reject } H_0)$} \\[2pt]
\multirow{2}{*}{Annual accrual} & \multicolumn{5}{c}{$\theta_{\mathrm{truth}}$} \\
\cmidrule(lr){2-6}
 & 0.62 & 0.70 & 0.75 & 0.80 & 0.85 \\
\midrule
2.5 / yr  & 0.0489 & 0.3277 & 0.6436 & 0.8798 & 0.9716 \\
5.0 / yr  & 0.0482 & 0.2856 & 0.5644 & 0.8123 & 0.9445 \\
10.0 / yr & 0.0497 & 0.2549 & 0.4902 & 0.7345 & 0.8950 \\
\midrule
\multicolumn{6}{l}{\textbf{Panel B.} Timing and ethical burden} \\[2pt]
Annual accrual & $E[T_{\IA}]$ & $E[T_{\FA}]$ & $\Bw$ & $\Fw$ & \\
\midrule
2.5 / yr  & 60.63 & 139.44 & 0.62 & 0.67 & \\
5.0 / yr  & 31.82 &  81.79 & 1.25 & 1.04 & \\
10.0 / yr & 17.41 &  52.90 & 2.50 & 1.34 & \\
\bottomrule
\end{tabular}
\end{table}

This sensitivity analysis evaluates the robustness of the calibrated
WCR designs under joint operational and clinical uncertainty. The
accrual rate determines how calendar time translates into follow-up
maturity, whereas the true survival probability determines the event
generation process and therefore the rate at which inferential
information accumulates. Thus, the grid in Tables~\ref{tab:1B-lm-ess}
and~\ref{tab:1B-hr-ess} stress-tests the information--time alignment of
the fixed designs when the realized trial trajectory deviates from the
assumptions used at calibration.

The landmark design absorbs accrual perturbation relatively well. The
post-lock maturity window $X$ preserves a minimum amount of follow-up in
the locked cohort, so the information available for the landmark
estimand remains comparatively stable across accrual rates. At the
design alternative $\theta_{\mathrm{truth}} = 0.80$, power decreases
from $0.816$ at 2.5 patients per year to $0.779$ at 10 patients per
year, while type I error remains close to the nominal level across the
grid. The only practically meaningful departure is the fast-accrual
scenario, where power falls slightly below the 0.80 target. This suggests
that if accrual proceeds substantially faster than the design assumption, a
modest upward adjustment of $N_T$ or recalibration under the
anticipated accrual rate may be warranted.

The HR design is more sensitive to fast accrual. Although the HR
estimand uses unrestricted follow-up, the amount of unrestricted
follow-up available by the final analysis still depends strongly on the
calendar-time trajectory of enrollment. Faster accrual compresses the
interval over which follow-up can mature before analysis, thereby
reducing the effective information accumulated under the HR likelihood.
Consequently, at $\theta_{\mathrm{truth}} = 0.80$, power decreases from
$0.880$ at 2.5 patients per year to $0.735$ at 10 patients per year,
whereas type I error remains controlled across all three accrual rates.
This pattern illustrates an important distinction between estimand
structure and information availability: unrestricted follow-up can
increase efficiency when sufficient calendar time is available, but it
does not eliminate sensitivity to rapid enrollment when follow-up
maturity is compressed.

The ethical-burden metrics behave as operational quantities rather than
efficacy quantities. For both ESS-optimal designs, the waiting-window
length is fixed at $X=3$, so the expected number of patients enrolled
during the waiting window scales approximately as $rX$, yielding
$\Bw = 0.62$, $1.24$, and $2.49$ across the three accrual rates. The
mean per-patient waiting-window follow-up $\Fw$ is bounded by the $X/2$ heuristic in the absence of event censoring and truncation, reflecting the average exposure time for patients enrolled
uniformly within the waiting window before the interim decision. In the simulated rows with $X=3$, event censoring and finite waiting-window truncation yield empirical values below 1.5 months. These
metrics are essentially invariant to $\theta_{\mathrm{truth}}$ but
highly responsive to accrual intensity, reinforcing their interpretation
as measures of operational and ethical exposure induced by the interim
timing mechanism itself.



\subsection{Comparison with Existing Single-Arm Phase~II Designs}
\label{sec:comparison}

We compare WCR with four established single-arm Phase~II designs for
time-to-event or landmark endpoints: BayesDesign \citep{wu2021bayesian},
OneArm2stage/OSLRT \citep{wu2020twostage}, Simon's two-stage design
\citep{simon1989optimal}, and BOP2-TTE \citep{zhou2020bop2tte}.  All five
methods are evaluated under the single-Weibull/exponential calibration baseline used for the Section~5 comparison:
$\theta_0 = 0.62$, $\theta_{\mathrm{alt}} = 0.80$, an exponential
reference $S_0$ satisfying $S_0(24) = 0.62$, Poisson accrual at
5 patients per year, one-sided $\alpha = 0.05$, target power 0.80,
and $10{,}000$ simulation replicates.  BayesDesign, OSLRT, and
Simon's design use the sample size returned by their native
calibration and meet the joint error-rate and power target.  BOP2-TTE
is the exception: for an aligned secondary comparison, we force
$N_T$ and $n_1$ to the WCR-HR values and then examine the consequence
for type~I error control in the HR subsection.  The main comparison is
therefore on maximum sample size, expected sample size under $H_0$,
and analysis timing, not on whether a method can be calibrated at all. Accordingly, BOP2-TTE is interpreted as an aligned but uncalibrated sensitivity row rather than as a calibrated efficiency comparator; the daggered table notes flag this distinction wherever BOP2-TTE appears.

The two estimand frameworks --- landmark $\theta = S(\tau)$ and
hazard-ratio $\delta = \log S(\tau)/\log S_0(\tau)$ --- test the same
null hypothesis $S(\tau) \le 0.62$, but they use observed follow-up
differently. The landmark analysis caps each patient's contribution
at the horizon $\tau = 24$ months, so the test statistic depends only
on whether each patient reaches $\tau$. The hazard-ratio analysis
uses the full follow-up of every patient and exploits the assumed
proportional-hazards structure $S(t) = S_0(t)^{\delta}$. As a result,
the HR framework requires substantially smaller sample sizes than the
landmark framework but completes its final analysis at a later
calendar time relative to enrollment.

\subsubsection{Landmark framework}
\label{sec:p2-landmark}

Four methods test
$H_0: S(\tau) \le 0.62$ vs.\ $H_1: S(\tau) > 0.62$. WCR uses the
ESS-optimal landmark design from Section~\ref{sec:1B}
($N_W = 19$, $X = 3$, $p_{\IA} = 0.310$, $p_{\FA} = 0.0406$,
$N_T = 45$). BayesDesign uses an exponential failure model with an
inverse-gamma prior on the hazard and an event-driven trigger: the
interim is taken when $m_1$ landmark events have been observed and
the final when $m$ events have accrued, with the calibration
selecting the smallest final event count $m$. OneArm2stage / OSLRT
applies the one-sample log-rank score under the
\texttt{Two\_stage\_Optimal} (minimum expected sample size) row, with
\texttt{restricted = 1} and follow-up $t_f = \tau$. Simon's binomial
design operates on the binary indicator $\mathbf{1}\{T_i > \tau\}$
with $p_0 = 0.62$, $p_1 = 0.80$; both interim and final analyses
wait $\tau$ months after the relevant enrollment milestone before the
responder count is observed, so $T_{\IA} = a_{(n_1)} + \tau$ and
$T_{\FA} = a_{(n)} + \tau$. None of the comparators is forced to
match WCR's $N_T$ or ESS in the landmark framework; each is calibrated
under its package-native optimization criterion.

\begin{table}[ht]
\centering
\caption{Landmark-estimand comparison ($\theta_0 = 0.62$,
  $\theta_{\mathrm{alt}} = 0.80$, $r = 5/12$, exponential $S_0$,
  $10{,}000$ replicates).  Sample sizes and calendar times are
  mean\,(sd) under $H_0$.}
\label{tab:p2-landmark}
\footnotesize
\setlength{\tabcolsep}{4pt}
\resizebox{\textwidth}{!}{%
\begin{tabular}{lp{5.8cm}ccclll}
\toprule
Method & Design parameters & $\alpha$ & Power & Realised $N|H_0$ & $E[T_{\IA}|H_0]$ & $E[T_{\FA}|H_0]$ \\
\midrule
WCR & $N_W=19$, $X=3$, $N_T=45$, $p_{\IA}=0.310$, $p_{\FA}=0.0406$ & 0.0432 & 0.810 & 29.6 (12.0) & 46.1 (10.2) & 128.8 (15.8) \\
BayesDesign (Wu) & n1=21, n=42, m1=6, m=12, eta=0.95, xi=0.08 & 0.0471 & 0.816 & 25.4 (12.3) & 46.8 (15.3) & 98.7 (21.0) \\
Simon two-stage & n1=14, r1=9, n=46, r=33 & 0.0492 & 0.807 & 31.4 (10.7) & 57.5 (8.9) & 133.8 (16.5) \\
OSLRT (OneArm2stage) & n1=26, n=46, c1=0.136, c=1.617 & 0.0406 & 0.831 & 35.0 (10.0) & 62.3 (12.1) & 134.3 (16.2) \\
\bottomrule
\end{tabular}
}
\end{table}

\subsubsection{Hazard-ratio framework}
\label{sec:p2-hr}

Four methods test $H_0: \delta \ge 1$ vs.\ $H_1: \delta < 1$, with
$\delta_{\mathrm{alt}} = \log(0.80)/\log(0.62) \approx 0.467$. WCR
uses the ESS-optimal HR design from Section~\ref{sec:1B} ($N_W = 13$, $X = 3$,
$p_{\IA} = 0.410$, $p_{\FA} = 0.0483$, $N_T = 25$). BayesDesign uses
the same exponential--inverse-gamma model with \texttt{method = "hr"}
(no $\tau$-cap on follow-up); its native event-driven calibration
selects $n = 24$. OSLRT is calibrated separately with
\texttt{restricted = 0}; its log-rank score is intrinsically
HR-flavoured and the calibration selects $n = 26$. BOP2-TTE places a
Gamma prior on the hazard $\lambda$ and triggers the interim at $N/2$
enrollments, the final at $N$, with the decision rule
$\Pr(\lambda < \lambda_0 \mid \text{data}_k) <
 C \cdot (n_k / N)^{\gamma}$. To put BOP2-TTE on the same trial-size
footing as the other three HR comparators, we deliberately
constrain it to $N_T = 25$ and $n_1 = 13$, a departure from the
package's native search (which freely chooses $N_T$ to satisfy the
calibration target). With $N_T$ fixed, we search $(C, \gamma)$ on
a grid of $1{,}470$ cells ($C \in [0.30, 0.99]$ in steps of $0.01$,
$\gamma \in [0, 2]$ in steps of $0.1$). No cell satisfies both
$\alpha \le 0.05$ and power $\ge 0.80$. The empirical
(Type~I, power) frontier visible on this grid does not pass through
the feasible rectangle: at the $\alpha$-conservative end the cell
$(C, \gamma) = (0.95, 0)$ gives $\alpha = 0.034$ but only $0.36$
power, while on the power-feasible boundary --- the subset of grid
cells with power $\ge 0.80$ --- the smallest achievable $\alpha$ is
$0.186$. The reported design is the $(C, \gamma) = (0.87, 1.6)$
cell that minimises the deviation from the joint target, giving
$\alpha = 0.166$ and power $0.79$. This is an empirical limitation of
the aligned $N_T = 25$ implementation, not a theorem about BOP2-TTE:
a finer grid, a different prior, or a different reference survival
might shift the frontier. Under the grid scanned here, however, the
two BOP2-TTE tuning parameters $(C, \gamma)$ are insufficient to meet
the four target operating characteristics induced by early and final
decisions under $H_0$ and $H_1$. By contrast, WCR uses
$(N_W, X, p_{\IA}, p_{\FA})$, OSLRT uses $(n_1, n, c_1, c)$, and
BayesDesign lets the event-driven enrollment target vary with the
observed data. A reader who
prefers an $\alpha$-controlled BOP2-TTE design should let the
package select its maximum sample larger than 30 under the same exponential
reference, at the cost of dropping the trial-size alignment with
the other three HR comparators.

\begin{table}[ht]
\centering
\caption{Hazard-ratio comparison ($10{,}000$ replicates,
  exponential $S_0$).  Mean\,(sd) under $H_0$.}
\label{tab:p2-hr}
\footnotesize
\setlength{\tabcolsep}{4pt}
\resizebox{\textwidth}{!}{%
\begin{tabular}{lp{5.8cm}ccclll}
\toprule
Method & Design parameters & $\alpha$ & Power &  $E(N|H_0)$ & $E[T_{\IA}|H_0]$ & $E[T_{\FA}|H_0]$ \\
\midrule
WCR & $N_W=13$, $X=3$, $N_T=25$, $p_{\IA}=0.410$, $p_{\FA}=0.0483$ & 0.0482 & 0.812 & 19.5 (5.4) & 31.8 (8.3) & 81.5 (11.7) \\
BayesDesign (Wu) & n1=12, n=24, m1=6, m=12, eta=0.95, xi=0.07 & 0.0477 & 0.815 & 20.4 (6.9) & 40.5 (9.8) & 67.0 (10.3) \\
BOP2-TTE$^{\dagger}$ & N=25, $n_{\mathrm{int}}$=13, C=0.87, $\gamma$=1.6 & 0.1657 & 0.786 & 22.1 (5.1) & 31.1 (8.7) & 53.0 (16.6) \\
OSLRT (OneArm2stage) & n1=16, n=26, c1=-0.290, c=1.613 & 0.0430 & 0.824 & 22.2 (4.8) & 38.4 (9.6) & 86.5 (12.2) \\
\bottomrule
\end{tabular}
}
\smallskip

{\footnotesize $^{\dagger}$No $(C,\gamma)$ combination satisfies
$\alpha \le 0.05$ and power $\ge 0.80$ at the aligned $N_T = 25$;
reported values are the grid cell closest to the joint target.}
\end{table}

\subsubsection{Sensitivity under the RMS-2021 piecewise Weibull DGM}
\label{sec:p2-pw-robustness}

The four comparison designs above use an exponential reference
$S_0(t) = \exp(-\lambda_0 t)$ anchored at $S_0(\tau{=}24) = 0.62$.
The RMS-2021 ARST0531 VAC arm, however, is not well described by an
exponential survival curve.  A piecewise Weibull fit to the
reconstructed Kaplan--Meier curve gives an increasing hazard during
the first $24$ months (shape $1.89$, scale $35.5$), followed by a
near-flat post-induction hazard (shape $0.079$, scale $169.5$), with
cumulative-hazard continuity at the $24$-month break point. The
piecewise fit and the exponential reference agree only at $\tau$
($S_0(24) = S_{\mathrm{PW}}(24) = 0.62$); they differ markedly
elsewhere. Under the piecewise model, half of all patients survive
past $120$ months, compared with $9\%$ under the exponential
reference. To assess whether each calibrated design retains nominal
type I error when the truth is the piecewise Weibull rather than the
exponential reference used at calibration, we run $10{,}000$
replicates per method and estimand under the piecewise Weibull DGM,
holding the design parameters fixed.

\begin{table}[ht]
\centering
\caption{Robustness under the RMS-2021 piecewise Weibull DGM
  ($10{,}000$ replicates).  Both DGMs satisfy
  $S(\tau{=}24) = 0.62$; PW parameters are given in
  Section~\ref{sec:p2-pw-robustness}.}
\label{tab:p2-pw-dgm}
\footnotesize
\setlength{\tabcolsep}{4pt}
\begin{tabular}{lcccccc}
\toprule
 & \multicolumn{2}{c}{Exp DGM (calibration)} & \multicolumn{4}{c}{PW DGM (RMS-2021)} \\
\cmidrule(lr){2-3} \cmidrule(lr){4-7}
Method & $\alpha$ & Power & $\alpha$ & Power & $E[T_{\IA}|H_0]$ & $E[T_{\FA}|H_0]$ \\
\midrule
\multicolumn{7}{l}{\emph{Landmark estimand}} \\
WCR & 0.0447 & 0.800 & 0.0744 & 0.868 & 46.4 (10.3) & 129.1 (15.9) \\
BayesDesign (Wu) & 0.0471 & 0.816 & 0.0472 & 0.815 & 50.7 (15.3) & 102.7 (20.9) \\
OSLRT (OneArm2stage) & 0.0406 & 0.833 & 0.1866 & 0.967 & 62.6 (12.3) & 134.5 (16.4) \\
Simon two-stage & 0.0492 & 0.807 & 0.0510 & 0.802 & 57.7 (9.0) & 134.3 (16.3) \\
\midrule
\multicolumn{7}{l}{\emph{Hazard-ratio estimand}} \\
WCR & 0.0486 & 0.811 & 0.5889 & 0.963 & 31.9 (8.4) & 81.7 (12.0) \\
BayesDesign (Wu) & 0.0477 & 0.815 & 0.0478 & 0.816 & 49.2 (14.1) & 94.4 (18.1) \\
OSLRT (OneArm2stage) & 0.0433 & 0.828 & 0.9576 & 0.999 & 38.6 (9.7) & 86.9 (12.4) \\
BOP2-TTE$^{\dagger}$ & 0.1657 & 0.786 & 0.2775 & 0.872 & 31.1 (8.7) & 57.1 (13.9) \\
\bottomrule
\end{tabular}
\smallskip

{\footnotesize $^{\dagger}$BOP2-TTE failed calibration at the aligned
$N_T = 25$ (see Table~\ref{tab:p2-hr}); its sensitivity results
should be interpreted with this caveat.}
\end{table}

The results in Table~\ref{tab:p2-pw-dgm} separate mostly according to
whether the method's statistic is invariant to the chosen $S_0$.
BayesDesign and Simon control type I error to within Monte
Carlo noise of the calibration target under both DGMs. By contrast, the HR-WCR type~I error of 0.589 under the piecewise DGM is a severe reference-curve sensitivity signal, not a minor robustness deviation: HR-WCR should be calibrated with a correctly specified reference curve over the observed follow-up range, and the landmark WCR version is the safer recommendation for RMS-like long-tail hazards when such calibration is uncertain.
BayesDesign's invariance is by construction: the design is
event-driven, with interim and final analyses triggered by
pre-specified event counts (e.g.\ $m_1 = 6$, $m = 12$) rather
than by enrollment milestones.  Because the decision boundaries
depend only on the accumulated information (number of events),
not on the calendar time or enrollment pace, the operating
characteristics are determined entirely by the event-count
targets and are invariant to the choice of $S_0$. Simon's invariance
follows because the decision uses only the binary indicator
$\mathbf{1}\{T_i > \tau\}$, whose probability under either DGM is
$S_0(\tau) = 0.62$ by construction.

The remaining designs use an exponential reference internally and
are not invariant.  We discuss each estimand in turn.

For landmark estimand, WCR inflates from $\alpha = 0.045$ to $0.074$, a manageable
departure.  The piecewise model's increasing pre-$\tau$ hazard
(shape$_1 = 1.89$) concentrates events near $\tau$ rather than
spreading them uniformly, so more patients in the locked cohort
have events close to the landmark, slightly biasing the posterior
toward rejection.  OSLRT inflates more severely to
$\alpha = 0.187$: its score statistic
$(E - O)/\sqrt{E}$ centers at zero only when the expected event
count $E = \sum -\log S_0(\mathrm{fu}_i)$ is computed under the
true survival, and the exponential reference systematically
underestimates $E$ when pre-$\tau$ hazard is increasing.

For hazard-ratio estimand, all three non-invariant HR methods inflate substantially.  OSLRT
essentially always rejects ($\alpha = 0.958$), because the same
$E-O$ bias is amplified when follow-up extends into the
near-flat post-$\tau$ regime (shape$_2 = 0.079$) where the
exponential reference predicts far more events than actually occur.

WCR inflates from $\alpha = 0.049$ to $0.589$.
Table~\ref{tab:p2-pw-diag} shows the mechanism: under the
piecewise DGM the average event count at the final analysis drops
to $9.7$ (from $15.2$ under the exponential reference), because
patients surviving past $\tau$ rarely fail under the near-flat
post-induction hazard.  Meanwhile the transformed times
$W_i = \lambda_0 T_i$ accumulate to a mean sum of $19.5$ (versus
$16.7$), reflecting that observed survival times are systematically
longer than the exponential reference predicts.  The posterior mean
of $\delta$ is therefore pulled to $\approx 0.50$, far below the
null $\delta = 1$, which the decision rule reads as strong evidence
for treatment effect.

\begin{table}[ht]
\centering
\caption{Per-trial diagnostics for the WCR-HR design ($N_T = 25$)
  under the two DGMs.  Means over $10{,}000$ replicates reaching
  the final analysis.}
\label{tab:p2-pw-diag}
\small
\begin{tabular}{lcc}
\toprule
Quantity & Exponential DGM & Piecewise-Weibull DGM \\
\midrule
$\bar D$ (events at FA)               & $15.2$ & $9.7$  \\
$\overline{\sum W_i}$ (transformed time at FA) & $16.7$ & $19.5$ \\
Posterior mean $\hat\delta$           & $0.91$ & $0.50$ \\
Empirical $\alpha$                    & $0.049$ & $0.589$ \\
\bottomrule
\end{tabular}
\end{table}

BOP2-TTE inflates from $\alpha = 0.166$ to $0.278$ under the
piecewise Weibull DGM.  Because BOP2-TTE already failed calibration
at the aligned $N_T = 25$ (Table~\ref{tab:p2-hr}), this shift is
difficult to interpret: the misspecification compounds on top of an
uncalibrated baseline rather than starting from nominal $\alpha$
control.

The exponential working reference is therefore reasonably stable for
the landmark estimand but fragile for HR when the true post-$\tau$
hazard departs from a constant-hazard model.  For RMS-like long-tail
hazards, the practical recommendation is to use the landmark estimand,
as in the RMS2021 worked example in Section~4.  If an HR estimand is
required, WCR should be recalibrated with the planner's best
specification of $S_0$.  This recalibration is already supported by
the package: \texttt{two\_stage\_optimize\_design()} and
\texttt{oc\_two\_stage()} accept user-supplied \texttt{S0\_fun} and
\texttt{S0\_inv\_fun} arguments, and Section~4 walks through that
workflow with the piecewise Weibull fit to ARST0531.  The comparison
study holds a single exponential reference across all eight cells to
isolate the inferential machinery of each method from differences in
the assumed null model, not to recommend the exponential reference for
an RMS-shaped trial.

\subsubsection{Trial duration under the alternative hypothesis}
\label{sec:p2-tfa-h1}

The summary statistics in Tables~\ref{tab:p2-landmark} and
\ref{tab:p2-hr} report the mean and standard deviation of the final
analysis time $T_{\FA}$ under $H_0$.  For trial planning, the
distribution under $H_1$ is at least as important: a sponsor needs to
know how long the trial will take when the treatment works, not only
when it does not.  We therefore re-run each calibrated design with
simulated event times drawn under $S_{\mathrm{alt}}(\tau) = 0.80$ (or $\delta_{\mathrm{alt}} = 0.467$ in HR estimand) and report the
empirical distribution of $T_{\FA}$ conditional on the trial reaching
the final analysis. Each panel of
Figures~\ref{fig:tfa-h1-lm} and \ref{fig:tfa-h1-hr}
overlays the histogram with a red dashed line at the
\emph{design-time reference} for $T_{\FA}$: $n/r + t_f$ for the
enrollment-driven designs (WCR, OSLRT, Simon, BOP2-TTE) and
$n/r + \tau$ for Wu's event-driven design. Here $n$ is the
design-time maximum sample size, $r = 5/12$, and $t_f$ or $\tau$
is the post-enrollment follow-up duration.

All eight cells share a common Poisson arrival stream at $r = 5/12$,
so that calendar-time variability is comparable across methods.  This
choice deliberately differs from the conditional-Uniform assumption
used by OneArm2stage's \texttt{phase2.TTE()} for $\alpha$/power
calibration; conditioning on Uniform arrival times fixes the accrual
duration at $n/r$ and suppresses the timing variability that any real
trial inherits from the arrival process.  Wu's BayesDesign uses the
same Poisson arrival stream as WCR in this comparison, but its trigger
is event-driven: enrollment continues until $m$ events accrue rather
than stopping at $n$ patients.

\begin{figure}[ht]
\centering
\includegraphics[width=0.95\textwidth]{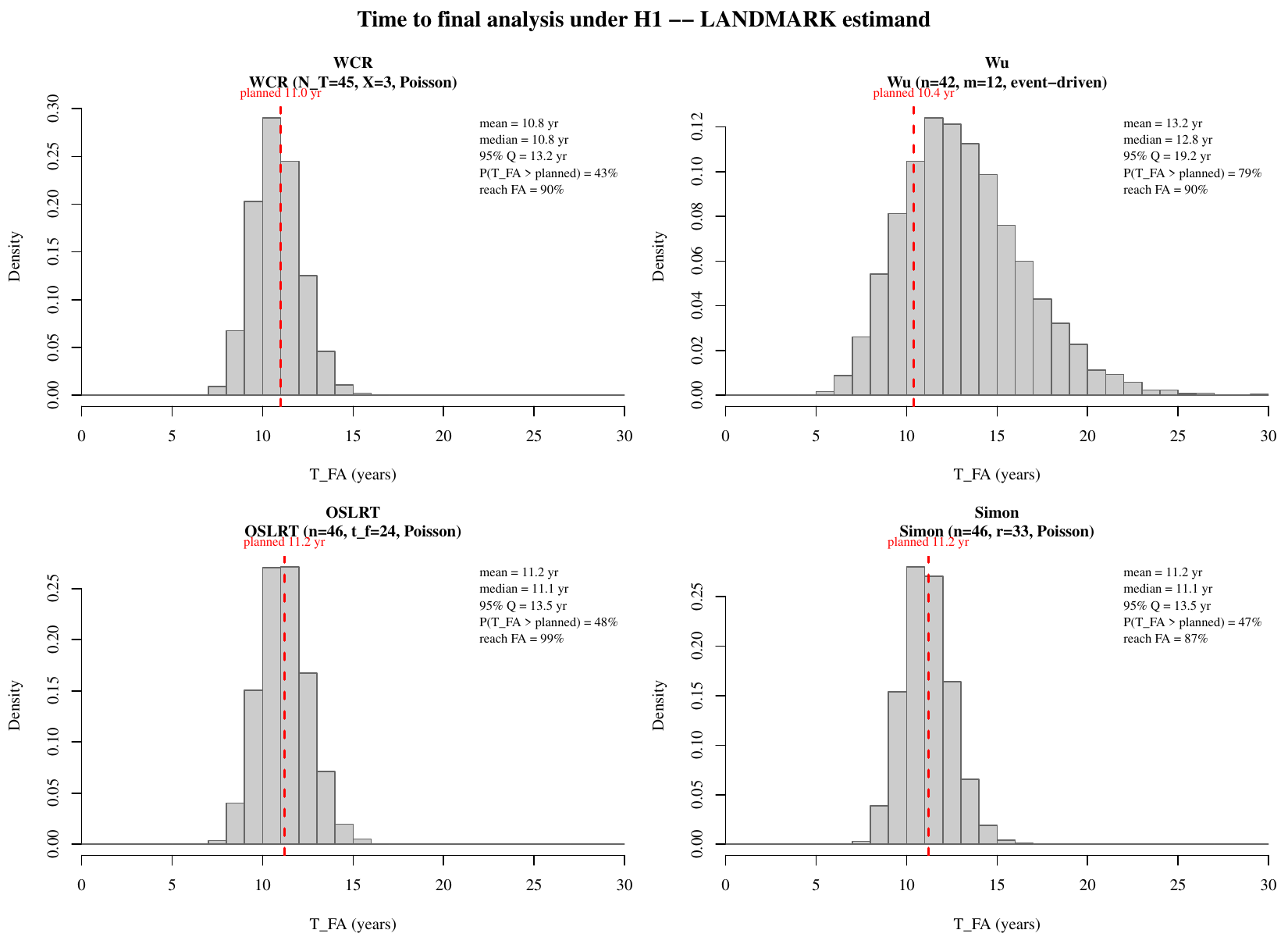}
\caption{Distribution of the final-analysis calendar time $T_{\FA}$
under $H_1$ for the four landmark designs, simulated under a common
Poisson accrual stream at $r = 5/12$ patients per month. Each panel
shows the histogram of $T_{\FA}$ across $10{,}000$ simulated trials,
conditional on reaching the final analysis (the percentage that
does so is reported in each panel). The red dashed line marks the
\emph{design-time reference} for $T_{\FA}$, defined as $n/r + t_f$
for the enrollment-driven designs (WCR, OSLRT, Simon) and as
$n/r + \tau$ for Wu's event-driven design. For the event-driven
case this reference is not a fixed calendar duration --- the trial
runs until $m$ events accrue, with no calendar cap unless the
simulator's $200$-patient ceiling is hit --- but is shown to
anchor the exceedance probability against a directly comparable
quantity.}
\label{fig:tfa-h1-lm}
\end{figure}

\begin{figure}[ht]
\centering
\includegraphics[width=0.95\textwidth]{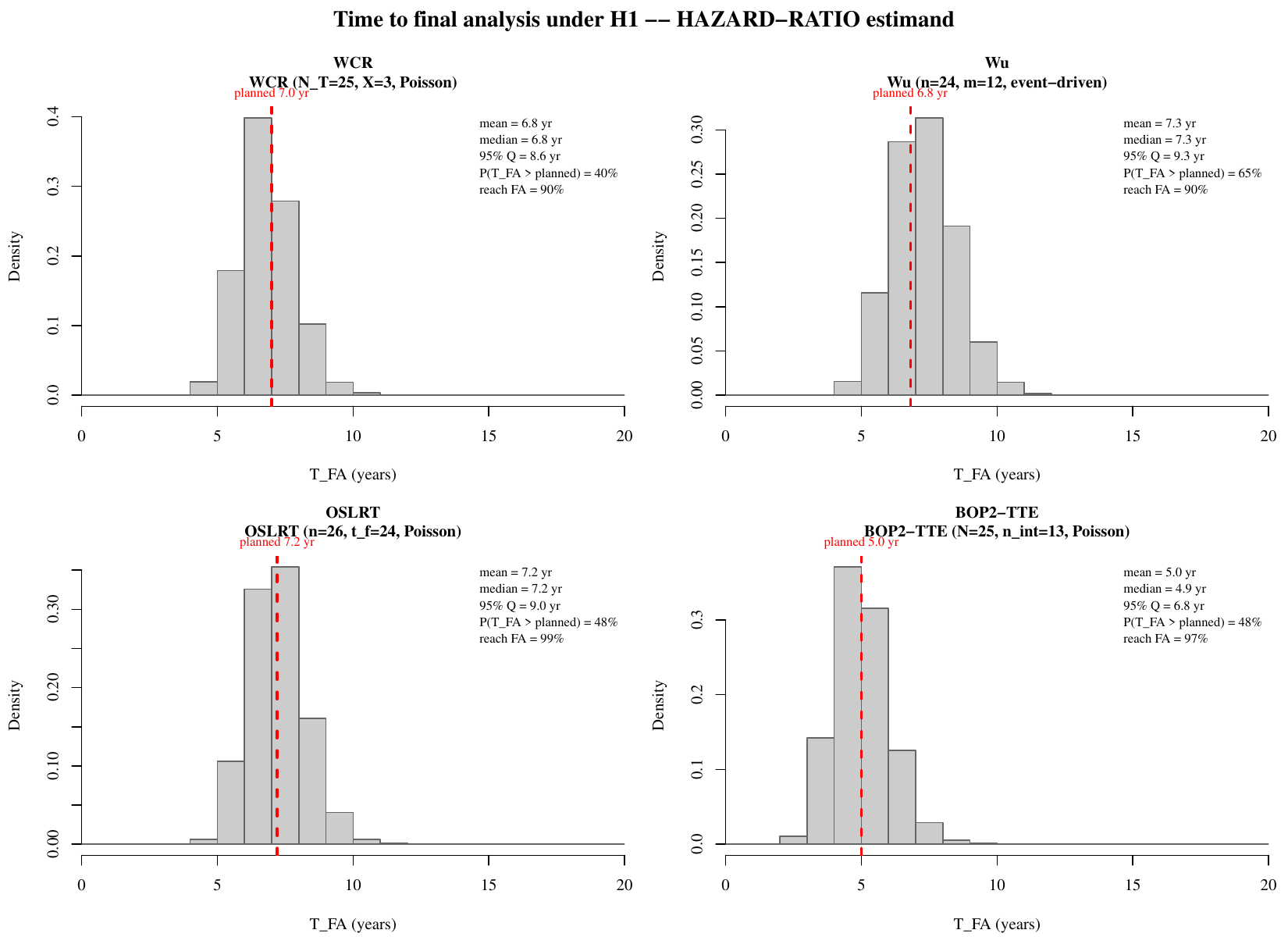}
\caption{Distribution of the final-analysis calendar time $T_{\FA}$
under $H_1$ for the four hazard-ratio designs. Same conventions
as Figure~\ref{fig:tfa-h1-lm}. BOP2-TTE here is the
$N_T = 25$ aligned calibration whose $\alpha = 0.166$ is not
controlled at the $0.05$ target; its shorter mean $T_{\FA}$ should
not be read as an efficiency advantage, since the comparison is
made between an $\alpha$-controlled design (WCR, OSLRT,
BayesDesign) and one that is not.}
\label{fig:tfa-h1-hr}
\end{figure}

The $T_{\FA}$ distributions fall into two patterns, determined by
what triggers the final analysis.  Enrollment-driven designs (WCR,
OSLRT, Simon, BOP2-TTE) close enrollment at a fixed $N_T$ and then
wait a deterministic follow-up period $t_f$, so the trial duration
is $T_{\FA} = a_{(N_T)} + t_f$, where $a_{(N_T)}$ is the arrival
time of the last patient.  Under Poisson accrual at rate $r$, the
only calendar-time variability comes from $a_{(N_T)}$, which has
standard deviation $\sqrt{N_T}/r \approx 12$--$16$ months for the
sample sizes in this comparison.  Event-driven designs (Wu's
BayesDesign) instead trigger the final analysis when $m$ events
have been observed, regardless of how many patients have been
enrolled.  Because the number of patients needed to accumulate $m$
events is itself random, the calendar time inherits an additional
source of variability that enrollment-driven designs do not have.

The enrollment-driven pattern is visible in all six cells where it
applies: the empirical mean $T_{\FA}$ falls within $0.2$ years of
the design-time reference $N_T/r + t_f$, and 45--48\% of trials
run slightly past it --- a modest and predictable spread driven
entirely by the Poisson arrival process.

Wu's BayesDesign departs from this pattern because the final
analysis requires $m$ events rather than $N_T$ enrollments.  When
the treatment is effective ($H_1$), fewer patients experience
events, so the trial must enroll beyond the design-time $n$ to
reach the event target.  Under the landmark estimand with
$S_{\mathrm{alt}}(\tau) = 0.80$, the per-patient event probability
is only $1 - 0.80 = 0.20$, so reaching $m = 12$ events requires on
average $12/0.20 = 60$ patients --- 43\% more than the design-time
$n = 42$.  This over-enrollment shifts the empirical mean $T_{\FA}$
to $13.2$ years (versus the design-time reference of $10.4$ years)
and produces a heavy right tail: the 95\% quantile reaches $19.2$
years and the worst observed trial ran for $29.5$ years.  Under
the HR estimand the same mechanism operates in milder form, because
the event probability under $H_1$ is higher in HR mode, yielding a
mean $T_{\FA}$ of $7.3$ years against a design-time reference of
$6.8$ years.

\begin{table}[ht]
\centering
\caption{Alternative-hypothesis final-analysis tail summaries for the event-driven comparator rows discussed in text. The reference column is the design-time duration used for the exceedance calculation.}
\label{tab:event-driven-tail}
\footnotesize
\begin{tabular}{lccccccc}
\toprule
Estimand & Reference & Mean & SD & Median & 95th pct. & Max & $\Pr(T_{\FA}>\text{ref.})$ \\
\midrule
Landmark BayesDesign & 10.4 yr & 13.2 yr & 3.5 yr & 12.6 yr & 19.2 yr & 29.5 yr & 0.76 \\
HR BayesDesign & 6.8 yr & 7.3 yr & 1.2 yr & 7.3 yr & 9.5 yr & 12.7 yr & 0.68 \\
\bottomrule
\end{tabular}
\end{table}

These are stress-test tails, not forecasts for an actual protocol: no
DSMB stop, sponsor cap, or calendar cutoff is modeled apart from the
simulator's 200-patient ceiling.  The point is that an event-driven
trigger carries a much wider $T_{\FA}$ distribution than an
enrollment-driven trigger at the same nominal $\alpha$ and power.  A
sponsor using an event-driven design should account for that width at
the protocol stage, either through an explicit maximum trial length, a
low-event contingency rule, or an enrollment-driven trigger.  Under
WCR-landmark, 43\% of $H_1$ trials run past the design-time reference
of 11 years, with the worst at 17 years.  Under Wu-landmark, 79\%
exceed the 10.4-year reference, with the long tail running into the
simulation cap.  A hard calendar window, whether contractual,
regulatory, or driven by a finite patient pool, is easier to defend
with the enrollment-driven profile.

\subsubsection{Summary}

Seven of the eight evaluated designs achieve type I error within
$\pm 0.01$ of the nominal $0.05$ target and power within
$\pm 0.02$ of $0.80$ on independent verification; the exception is
BOP2-TTE in the hazard-ratio framework, whose discrete spending-
function calibration cannot deliver $\alpha = 0.05$ at the
$N_T = 25$ trial size shared with the other three HR methods (see
the dedicated discussion in the HR-framework subsection). For the
other seven designs the substantive comparison reduces to how each
delivers those error rates: the maximum sample size $N_T$ it
requires, the expected sample size under $H_0$, and the calendar
timing of the interim and final analyses.

The standard deviations of realised sample size and analysis times
should be interpreted with care, since they reflect different sources
of variability across methods. Accrual is simulated as a Poisson
process at $5$ patients per year, so even strictly enrollment-driven
analyses inherit a calendar SD of approximately $\sqrt{N_W} / r$
months from the inter-arrival times alone. The landmark interim SDs
of $10.3$, $8.9$, and $12.1$ months for WCR, Simon, and OSLRT are
essentially these accrual-driven values; under deterministic accrual
they would all collapse toward zero.

The realised sample size $N$ is variable for every method, but for
different reasons. WCR, Simon, OSLRT, and BOP2-TTE follow a two-point
distribution: $N$ equals the interim-stage size $n_1$ when the trial
stops at the interim and the maximum size $n$ when it continues. For
these designs the SD is therefore bounded by the difference $n - n_1$
weighted by the early-stop probability. BayesDesign uses adaptive
enrollment instead --- it keeps recruiting patients until the target
$m$ events are observable --- so the realised $N$ varies continuously
trial-by-trial. Empirically the BayesDesign HR design carries a
realised $N$ of $20.4$ patients on average with SD $6.9$ and range
$[9, 43]$ across $10{,}000$ replicates, which is wider than WCR's
two-point spread of mean $19.5$ with SD $5.4$ over the support
$\{N_W, N_T\} = \{13, 25\}$.

This adaptive-enrollment mechanism is also what produces the
otherwise surprising HR result that BayesDesign's $T_{\FA}$ SD
($10.3$ months) is smaller than WCR's ($12.0$ months). Under WCR the
final analysis is anchored at $a_{(N_T)} + \tau$ and inherits the
full $\sqrt{N_T}/r \approx 12$-month accrual SD. Under BayesDesign,
slow accrual is automatically compensated by recruiting more patients,
which restores the cumulative exposure needed to observe the target
$m$ events; the negative feedback between accrual and trial duration
trims the calendar SD at the cost of a wider realised-$N$
distribution. The clearest contrast remains the landmark BayesDesign
comparator, whose $T_{\FA}$ SD of $21.0$ months exceeds the $\sim 16$
months shared by WCR, Simon, and OSLRT --- in the landmark setting
the adaptive-enrollment compensation is weaker because $m = 12$
events are needed and per-patient information is capped at $\tau$.

Under the landmark estimand, BayesDesign delivers the smallest mean
realised sample size ($25.4$ patients on average) by enrolling only
until the target $m = 12$ landmark events become observable. Its
expected final-analysis time of $98.7$ months is roughly thirty
months earlier than WCR's $128.8$ months, but its $T_{\FA}$ SD is
also $30\%$ larger ($21.0$ versus $15.8$ months) and its realised
$N$ varies continuously rather than between two values. WCR sits
between BayesDesign and the two enrollment-driven comparators on
sample-size mean: $29.6$ versus Simon's $31.4$ and OSLRT's $35.0$.
WCR also has the earliest mean interim ($46.1$ months versus
$62.3$ for OSLRT and $57.5$ for Simon). The landmark choice is
therefore a trade-off between the smaller mean realised $N$ and
earlier completion of BayesDesign on the one hand, and the bounded
realised $N$ and predictable timing of WCR on the other.

Under the hazard-ratio estimand, restricting attention to the three
designs that achieve $\alpha = 0.05$, WCR has the smallest mean
realised sample size ($19.5$ versus $20.4$ for BayesDesign and
$22.2$ for OSLRT) and the earliest interim analysis
($31.8$ months versus $40.5$ for BayesDesign and $38.4$ for OSLRT).
It does not dominate on every dimension. BayesDesign uses a slightly
smaller maximum sample size ($n = 24$ versus $N_T = 25$ for WCR) and
delivers a substantially earlier final analysis ($67.0$ versus
$81.7$ months) with a tighter $T_{\FA}$ distribution ($10.3$ versus
$12.0$ months SD), at the cost of a later interim look and a
continuous realised-$N$ distribution that ranges from $9$ to $43$
patients in our simulations. OSLRT achieves timing comparable to
WCR with $n = 26$. BOP2-TTE at the aligned $N_T = 25$ has
$\alpha = 0.166$ and is therefore not directly comparable on the
sample-size or timing axes; even setting the calibration gap aside,
its mean realised $N$ ($22.1$) and interim time ($31.1$ months) are
similar to WCR's, but the $\alpha$ inflation means any apparent
efficiency advantage is being purchased by an unacceptable type I
error rate.

The comparison does not produce a universal winner. BayesDesign
buys an earlier mean completion and a smaller mean realised $N$
through adaptive event-driven enrollment, at the cost of a
realised $N$ that varies trial-by-trial and a wider $T_{\FA}$ tail
under landmark. WCR buys a fixed enrollment budget and the
earliest interim under HR by pinning the locked cohort and the
waiting window. Among the four HR comparators, BOP2-TTE was the
only one that could not be brought to $\alpha = 0.05$ at the
aligned $N_T = 25$ on the $1{,}470$-cell $(C, \gamma)$ grid we
scanned; in this aligned grid, two tuning parameters were too coarse
for four operating-characteristic targets at this cohort size. The
choice between WCR and
BayesDesign in HR is the choice between a budget cap on
enrollment and a tighter completion-time distribution.

\section{Discussion}

The proposed Window-Cohort with Calibrated Follow-Up Requirement (WCR) framework was motivated by a structural limitation that arises repeatedly in modern single-arm time-to-event (TTE) trials, particularly in rare pediatric oncology settings with slow accrual, sparse events, and long-horizon endpoints. Existing interim monitoring paradigms typically rely on either event accumulation or enrollment progression to determine interim timing. Event-driven designs prioritize statistical information through observed events but often sacrifice calendar predictability, whereas enrollment-driven designs provide operational predictability but may trigger analyses based on immature follow-up. The central premise of this work is that these limitations are not merely operational inconveniences, but reflect a deeper conceptual issue regarding how information is represented and controlled in TTE interim monitoring.

A central contribution of the WCR framework is the recognition that existing paradigms typically control only proxies for information rather than information maturity itself. Event counts and enrollment counts are commonly treated as surrogates for inferential information, yet their relationship to actual information maturity becomes increasingly unstable in long-horizon and sparse-event settings. In particular, the same enrollment milestone may correspond to dramatically different levels of follow-up maturity depending on realized accrual dynamics, while the same event threshold may be reached only after prolonged calendar delays that render interim decisions operationally irrelevant. In this sense, traditional interim timing mechanisms implicitly treat timing as an emergent consequence of event realization or enrollment progression rather than as an explicit statistical design object.

The proposed framework instead treats follow-up maturity as a primary design parameter. Under the WCR construction, interim timing is jointly determined by a locked cohort size $N_W$ and a calibrated follow-up requirement $X$. This ``fix-then-wait'' structure directly controls the maturity of the interim dataset rather than relying on stochastic event realization or enrollment milestones alone. Conceptually, this introduces a third monitoring paradigm distinct from both event-driven and enrollment-driven designs. Unlike event-driven approaches, WCR does not wait for stochastic event accumulation to determine when decisions may occur. Unlike enrollment-driven approaches, WCR does not assume that enrollment alone determines the informativeness of interim data. Instead, follow-up maturity itself becomes the primary quantity being controlled.

More broadly, the current work suggests that follow-up maturity may represent the fundamental carrier of inferential information in TTE settings. Event counts and sample size become meaningful only insofar as they reflect the underlying maturity structure of follow-up. This distinction becomes particularly important in trials with long landmark horizons, delayed treatment effects, low event rates, or substantial administrative censoring. In such settings, the relationship between traditional proxies and actual inferential maturity may become highly unstable. From this perspective, the proposed framework can be viewed as a conceptual shift away from proxy-based interim timing toward direct maturity-based monitoring.

The framework also highlights that information maturity is inherently estimand-dependent. Under landmark survival estimands, effective information is naturally restricted to the clinically meaningful horizon $\tau$, whereas proportional hazards--based estimands induce unrestricted follow-up usage extending to the entire available follow-up window. Consequently, the same interim calendar time may correspond to fundamentally different effective information sets under different estimands. This observation suggests that interim timing cannot be separated from estimand specification. Rather, interim monitoring rules should be viewed as estimand-aligned inferential mechanisms rather than merely operational milestones. In this sense, the current work contributes to a broader ongoing transition toward estimand-aware clinical trial design and monitoring.

The implications of this distinction become particularly pronounced in rare pediatric oncology. In many contemporary pediatric solid-tumor trials, accrual may proceed slowly while clinically meaningful endpoints require one, two, or even several years of follow-up. Under such conditions, event-driven monitoring may become operationally disconnected from actionable decision-making because sufficient events may not accumulate until after a large fraction of patients have already been enrolled. Conversely, enrollment-driven monitoring may generate analyses dominated by immature follow-up and heavy administrative censoring. The resulting mismatch between operational timing and inferential maturity is not incidental but structural. The WCR framework attempts to address this mismatch by jointly controlling the timing and maturity of interim analyses through explicit follow-up parameterization.

An additional contribution of the proposed framework is the explicit characterization of ethical burden associated with interim timing mechanisms. In conventional monitoring paradigms, prolonged waiting periods prior to interim decisions are often treated as unavoidable operational consequences. However, under sparse-event settings, these waiting periods may correspond to substantial additional patient exposure before a futility decision can be made. The WCR framework formalizes this phenomenon through quantities such as the decision-lag enrollment burden and cumulative follow-up exposure during the waiting window. This perspective reframes ethical burden as a structural property of interim timing mechanisms rather than as an incidental consequence of trial conduct. By parameterizing the waiting window directly, the framework allows trade-offs among inferential maturity, operational feasibility, and patient exposure to become explicit and quantifiable components of trial design.

Importantly, the primary contribution of the current work is not inherently Bayesian. Although Bayesian posterior monitoring provides a convenient inferential engine within the current implementation, the conceptual contribution of WCR lies in the explicit parameterization of follow-up maturity and the joint alignment of information and calendar time. The framework is therefore compatible with alternative inferential paradigms, including likelihood-based approaches, frequentist monitoring procedures, hybrid Bayesian--frequentist methods, and decision-theoretic adaptive designs. In this sense, WCR should be interpreted primarily as an information-time framework rather than as a specific Bayesian monitoring design.

The current work also raises several broader theoretical questions. One direction involves formal characterization of effective information under administrative follow-up constraints. Another concerns asymptotic relationships between event-driven and maturity-driven monitoring regimes under sparse-event processes. A third involves optimality theory balancing inferential maturity, operational efficiency, and ethical burden under constrained accrual environments. More generally, the framework raises foundational questions regarding what constitutes information in TTE interim monitoring and how information should be coupled to calendar time in modern adaptive trial settings.

Another potentially important implication of the proposed framework
concerns the operational use of early follow-up horizons in
long-horizon TTE trials. In some clinical settings, investigators may
believe that an earlier milestone, such as 6-month PFS, is clinically
or biologically informative for a later primary horizon, such as
24-month PFS. Conventional approaches to incorporating such early
information often require explicit cross-horizon modeling, including
specification of the joint distribution or correlation structure
between short-term and long-term binary indicators. In contrast, the
WCR framework can operationally leverage early follow-up horizons
without redefining the primary estimand or requiring explicit
cross-milestone correlation modeling. For example, setting the waiting
window to $X=6$ months ensures that the locked cohort has accrued a
minimum amount of short-term follow-up prior to interim analysis,
while inference and decision-making may still target the primary
landmark estimand $\theta=S(\tau)$.

Importantly, WCR does not treat the earlier horizon as a formally
validated surrogate endpoint replacing the primary endpoint. Rather,
the framework treats early follow-up as partial time-to-event
information whose contribution is calibrated directly through the
operating characteristics of the monitoring procedure. In this sense,
the proposed approach may reduce implementation complexity and avoid
potential misspecification associated with explicit correlation-based
frameworks linking quantities such as $I(T>6)$ and $I(T>\tau)$.
However, calibration remains dependent on the assumed event-time
distribution and the estimand-specific information structure used by
the design. More formal comparisons between WCR and correlation-based
short-term monitoring approaches represent an important direction for
future methodological investigation.

Several limitations should also be acknowledged. First, the current implementation focuses primarily on single-arm phase II settings. Extension to randomized trials, multi-arm platforms, and adaptive enrichment settings remains an important area for future work.

Second, while the proposed framework explicitly controls minimum follow-up maturity, the relationship between maturity and inferential precision may still depend on the underlying event process and estimand structure. The sensitivity analyses show that this dependence is especially important for HR-based unrestricted follow-up, where reference-curve misspecification can severely distort type~I error.

Third, the current calibration strategy relies primarily on simulation-based optimization; additional theoretical characterization of optimality properties may further strengthen the framework. Protocol-specific applications should therefore report calibration seeds, simulation error, and sensitivity analyses under plausible alternative baseline curves and accrual rates.

Fourth, the proposed approach assumes that follow-up maturity can be meaningfully parameterized through a single waiting-window quantity $X$, whereas more complex dynamic or patient-specific maturity structures may arise in practice. The main operating-characteristic simulations also do not include independent dropout or informative censoring; these mechanisms should be added when they are plausible for the target disease setting.

At a broader level, the proposed framework reflects an evolving view of interim monitoring in modern clinical trials. Historically, interim timing was often inherited from event processes or enrollment progression and treated primarily as an operational scheduling issue. However, as trials increasingly incorporate long-horizon endpoints, sparse-event populations, adaptive monitoring, and estimand-driven inference, the coupling between information and time becomes progressively more central to statistical design itself. In this context, interim timing may no longer be viewed merely as a logistical consequence of accrual or event accumulation. Instead, interim timing may need to be treated as an estimand-aligned statistical design object that jointly governs information maturity, operational feasibility, and ethical burden.

More fundamentally, the current work suggests that modern TTE interim monitoring may benefit from moving beyond traditional proxy-based paradigms toward explicit maturity-based frameworks. The WCR design represents one possible step in this direction by directly parameterizing follow-up maturity and integrating information and time within a unified design structure. We hope that this perspective contributes to a broader rethinking of how interim monitoring should be conceptualized in rare-disease and long-horizon clinical trials.

\clearpage
\bibliographystyle{apalike}
\bibliography{references}

\end{document}